\documentclass{IEEEtran}
\IEEEoverridecommandlockouts
\usepackage{cite}
\usepackage{amsmath,amssymb,amsfonts,bm,amsthm}
\usepackage{algorithmic}
\usepackage{graphicx}
\usepackage{textcomp}
\usepackage{multirow}
\usepackage{xcolor}
\def\BibTeX{{\rm B\kern-.05em{\sc i\kern-.025em b}\kern-.08em
    T\kern-.1667em\lower.7ex\hbox{E}\kern-.125emX}}

\newcommand{\dv}[1]{{\boldsymbol{#1}}}
\newcommand{\dmt}[1]{\MakeUppercase{\boldsymbol{#1}}}
\newcommand{\rs}[1]{ \mathsf{#1}} 
\newcommand{\rv}[1]{\bm{{\mathsf{#1}}} }

\usepackage{url}

\usepackage{booktabs}

\newtheorem{proposition}{Proposition}
\newtheorem{remark}{Remark}

\usepackage{hyperref}

\usepackage{algorithm}
\usepackage{float}

\usepackage{lipsum}

\usepackage{relsize}

\begin{document}

\title{Rare-Event Simulation of Outage Probability in GSC/MRC Systems under Rician Fading
}

\author{
\IEEEauthorblockN{
Mahmoud Ghazal\IEEEauthorrefmark{1},
Nadhir Ben Rached\IEEEauthorrefmark{2},
and Tareq Al-Naffouri\IEEEauthorrefmark{1}
}

\IEEEauthorblockA{\IEEEauthorrefmark{1}Computer, Electrical and Mathematical Science and Engineering Division,\\
King Abdullah University of Science and Technology (KAUST), Thuwal, Saudi Arabia\\
Email: {mahmoud.ghazal@kaust.edu.sa}, {tareq.alnaffouri@kaust.edu.sa}\\
ORCID: \href{https://orcid.org/0009-0003-9827-3516}{0009-0003-9827-3516}, 
\href{https://orcid.org/0000-0001-6955-4720}{0000-0001-6955-4720}
}

\IEEEauthorblockA{\IEEEauthorrefmark{2}School of Mathematics, Faculty of Engineering and Physical Sciences,\\
University of Leeds, Leeds, United Kingdom\\
Email: {n.benrached@leeds.ac.uk}\\
ORCID: \href{https://orcid.org/0000-0003-4471-4075}{0000-0003-4471-4075}
}
}

\maketitle

\begin{abstract}
This paper studies the estimation of outage probability in GSC/MRC SIMO systems under Rician fading in the rare-event regime. The difficulty arises from the evaluation of the CDF of a partial sum of ordered non-central chi-square random variables, motivating the use of enhanced Monte-Carlo methods. For independent fading, we propose partition importance sampling (PIS), a theory-driven estimator tailored to this
structure, and prove that it achieves bounded relative error (BRE). We further adapt exponential twisting, proving its BRE property, and cross-entropy to this setting. We then extend ET and CE to correlated Rician fading, where the joint distribution of the power gains is no longer tractable, yielding the ETC and CEC estimators; ETC enjoys the bounded-relative-error guarantee under arbitrary mean and arbitrary covariance. Numerical experiments compare these methods with universal importance sampling and multilevel splitting for independent fading, and an asymptotic approximation in the correlated case. Empirically, CE shows the most robust performance in the independent case; PIS and ET are competitive but degrade for larger means, with ET further degrading when the selected subset is much smaller than the antenna array. ETC yields a better estimate than the asymptotic approximation for moderately rare events.
\end{abstract}



\section{Introduction}

    Several diversity schemes can be used to combine the inputs from multiple antennas in a single-input-multiple-output (SIMO) system. Notably, among them, generalized selection combining (GSC) allows for the choice of a smaller subset of antennas, usually with fixed size, in order to combine their inputs \cite{winAnalysisHybridSelectionMRCRician1999,alouiniMGFbasedAnalysisGSC2000}. The scheme is called so since it can be regarded as a generalized version of selection combining (SC), which uses only one antenna. Contrary to full combining schemes that traditionally utilize the whole set such as maximum ratio combining (MRC) and equal gain combining (EGC), GSC requires a smaller number of radio frequency (RF) chains, the use of which may be expensive \cite{sanayeiAntennaSelectionMIMOSystems}. To combine the signals from the smaller subset, MRC or EGC can be used. In our work, we are concerned with an MRC combination. We follow \cite{benrachedUniversalSplittingEstimator2020} and use the term GSC/MRC. 

    Performance analysis of GSC has attracted a number of researchers from the community. Fading regimes such as Rayleigh and Nakagami-m allow for closed-form expressions of outage probability (OP), as demonstrated in  \cite{namExactCaptureGSCind2013} and \cite{namNewClosedFormOrderStatisticsNakagami2017}. Finding OP boils down to evaluating the probability that a partial sum of order statistics is less than a given threshold. In many fading scenarios such as Rician, Weibull, and log-normal fading regimes, closed-form expressions are not available. Some researchers have resorted to MGF-based methods \cite{annamalaiUnifiedErrorProbaGSCRician2002,bithasPerformanceAnalysisGSCWeibull}. Another possible solution is to consider Monte-Carlo (MC) simulations. However, naive (crude, standard) MC  becomes impractical in the rare-event regime, since the required sample size scales inversely with the outage probability \cite{kroese2013handbook}. This has opened the door to a number of enhanced MC methods, as in \cite{benrachedSumOrderStatistics2018} and \cite{benrachedUniversalSplittingEstimator2020}.

    When it comes to the Rician fading regime, the literature on OP is limited to, to the best of our knowledge, MGF-based methods \cite{annamalaiUnifiedErrorProbaGSCRician2002,maEfficientPerformanceGSC2004,annamalaiAnalysisGeneralized2006}, and asymptotic cases \cite{songAsymptoticAnalysisGSCRician2013}. The MGF-based method starts with the MGF of the SNR, whose calculation requires a numerical integration scheme, and then inverts the MGF using another numerical integration scheme. It is also restricted to independent and identically distributed (i.i.d.) channel gains. This opens the door to a dedicated study of enhanced MC techniques of GSC diversity receivers under the Rician fading regime.


    In this paper, we (i) propose partition importance sampling (PIS), a new IS scheme that exploits the closure of independent non-central chi-square random variables under summation, and prove that it yields a bounded-relative-error (BRE) estimator of the OP; (ii) adapt the approximate exponential twisting (ET) scheme of \cite{benrachedEfficientImportanceSampling2021}, which was developed for full sums, to partial sums of order statistics, and prove that it is also a BRE estimator; (iii) adapt the cross-entropy (CE) scheme of \cite{deBoerTutorialCrossEntropy2005} to a parametric family of scaled non-central chi-squares, and give a sufficient condition for the BRE property; (iv) extend ET and CE to correlated Rician fading (ETC and CEC), where the joint distribution of the squared gains is no longer tractable, with ETC shown to have BRE; (v) benchmark the independent methods against universal importance sampling (UIS) \cite{benrachedSumOrderStatistics2018} and multi-level splitting (MLS) \cite{benrachedUniversalSplittingEstimator2020}; and (vi) compare the performance of ETC with a previously-established asymptotic approximation.
    

    The rest of the paper is structured as follows. Section \ref{sec:problem-formulation} formulates the problem. In Section \ref{sec:methods-independent}, structure-aware enhanced MC methods for the independent case are presented. Section \ref{sec:methods-dependent} shows those for the correlated case. Section \ref{sec:comparison-theoretical} compares the theoretical applicability and guarantees of the methods, and Section \ref{sec:comparison-numerical} utilizes numerical applications to compare the performance of the aforementioned methods. Finally, Section \ref{sec:conclusion} concludes the paper.

\section{Problem Formulation} \label{sec:problem-formulation}

    Consider a SIMO system with $M$ receive antennas, out of which only $m\leq M$ antennas can be utilized. The set of $m$ antennas is chosen so that the SNR is optimized. In case of known channel state information (CSI), this is equivalent to maximizing the sum of $m$ squared moduli of channel gains, i.e., applying MRC to the smaller set. Under Rician fading, the channel gains follow non-central complex normal distributions. We are interested in evaluating the OP for small thresholds, i.e., in the case of an outage being a rare event. Let $\rv{h}=[\rs{h}_1,\ldots,\rs{h}_M]^T$ be the vector of arbitrarily correlated complex normal channel gains: \begin{equation}
        \rv{h}\sim\mathcal{CN}(\dv{\mu},\dmt{\Sigma}). 
    \end{equation} Write $\dv{\mu}=[\mu_1,\ldots,\mu_M]^T$, and $\Sigma=(\rho_{ij}\sigma_i\sigma_j)_{1\leq i,j \leq M}$ with $\rho_{ii}=1$. Assume the covariance matrix $\dmt{\Sigma}$ is nonsingular.
    
    Let $\mathcal{S}$ be the sought event of outage, i.e., \begin{equation}
        \mathcal{S}=\left\{\max \sum_{k=1}^m |\rs{h}_{j_k}|^2 \leq \gamma_{\mathrm{th}}\right\},
    \end{equation} where maximization is over the $m$-subsets of indices $\{j_1,\ldots,j_m\}\subset \{1,2,\ldots,M\}$, and $\gamma_{\mathrm{th}}$ denotes the threshold (scaled) SNR. To simplify the notation, we will write \[\rs{X}_i=|\rs{h}_i|^2 \text{ for } i=1,2,\ldots,M.\] Sorting the squared moduli in decreasing order $$\rs{X}^{(1)}\geq \rs{X}^{(2)} \geq \ldots \geq \rs{X}^{(M)},$$ the outage event can be clearly rewritten as \footnote{The first expression presents the OP as the CDF of a vector of quadratic forms with diagonal matrices, while the second utilizes order statistics} \begin{equation}
        \mathcal{S}=\left\{H(\rv{X}):=\sum_{k=1}^m \rs{X}^{(k)} \leq \gamma_{\mathrm{th}}\right\}
    \end{equation} where \(\rv{X}=[\rs{X}_1,\ldots,\rs{X}_M]\). Let $P=\mathbb{P}(\mathcal{S})$ be the sought probability of outage.
    
    Evidently, we have \begin{equation}
        \rs{X}_i\sim\frac{\sigma_i^2}{2}\chi_2^2\left(\dfrac{2|\mu_i|^2}{\sigma_i^2}\right).
    \end{equation} The probability density function (PDF) of $\rs{X}_i$ is given by \[f_{\rs{X}_i}(x)=\dfrac{e^{-|\mu_i|^2/\sigma_i^2}}{\sigma_i^2}\exp\left(-\dfrac{x}{\sigma_i^2}\right)I_0\left(2\dfrac{|\mu_i|}{\sigma_i^2}\sqrt{x}\right),\; x>0.\] 
    
    The chosen channel distribution corresponds to a Rician fading scenario in which $\sigma_i^2$ is the power of the scattered paths, and $\nu_i^2=|\mu_i|^2$ is the power in the direct path, for $i=1,\ldots,M$. Note that the Rician fading has another representation, in which, the $K$-factors are $K_i=|\mu_i|^2/|\sigma_i|^2$, and total powers are $\Omega_i=|\mu_i|^2+\sigma_i^2$, again for $i=1,2,\ldots,M$.

    The outage event is rare when the preset threshold $\gamma_{\mathrm{th}}$ is small relative to the average total power per antenna. Fix the total number of antennas $M$ and the number of chosen antennas $m$. Since $\sigma_i^2$ is a scaling parameter, it can be normalized out, and the genuinely distinct regimes are: either (1) $\gamma_{\mathrm{th}}$ going to zero for fixed means, or (2) $\|\dv{\mu}\|$ going to infinity. In our work, we establish bounded relative error of some estimators in the former regime, and examine both in numerical experiments. The vanishing threshold regime corresponds to moderate $K$-factors: the dominance of the line-of-sight (LOS) component does not override the contribution of the non-LOS components.

\section{Estimation Methods: Independent Fading with Identical Scattering Power} \label{sec:methods-independent}

     In this section, we assume that the channels are independent and the scattering powers are identical. Without loss of generality, we can assume unit scattering powers: $\sigma_i^2=1$, for $i=1,\ldots,M$. Hence, we have \begin{equation}
        \rv{h}\sim\mathcal{CN}(\dv{\mu},\dmt{I}_M).
    \end{equation}

    Consider an unbiased estimator $\hat{P}$ of the probability $P$. Clearly, both are functions of the threshold $\gamma_{\mathrm{th}}$. As we are concerned with rare events, we are interested in the behaviour of the estimate as $\gamma_{\mathrm{th}}$ tends to zero. The \textit{estimated relative error} of the estimator can be defined as \begin{equation}
        \text{ERE}(\hat{P})=\dfrac{\sqrt{\operatorname{var}(\hat{P})}}{P}. \label{eq:ERE}
    \end{equation} If the estimator follows a normal distribution, the actual relative error, i.e., the absolute difference of the estimate and the real value divided by the latter, lies within $1.96\times\text{ERE}(\hat{P})$ with a probability of $95\%$. Note that, in general, if we drop the normality assumption, Chebyshev's inequality still guarantees that the actual relative error will be less than $4.47\times\text{ERE}(\hat{P})$, at least with a probability of $95\%$. Hence, enhancing standard MC requires feasibly reducing its variance.
    
    For instance, the central limit theorem (CLT) implies that the naive Monte-Carlo (NMC) estimator converges to a normal random variable. The variance of naive Monte-Carlo (NMC) is approximately equal to the sought small probability divided by the number of samples, hence the number of samples needed to achieve a predetermined accuracy is inversely proportional to the OP. For example, to estimate a probability $P=10^{-9}$ such the estimate is within a relative error of $5\%$ with a probability of $95\%$, one needs to generate $1.54 \times 10^{12}$ samples.
    
    Various performance metrics have been used in the rare event literature to investigate the effectiveness of an estimator \cite[Sec 10.1]{kroese2013handbook}. Among these criteria, a desired one is the \textit{bounded relative error}. If the estimate satisfies \begin{equation}
        \limsup_{\gamma_{\mathrm{th}}\to 0}\dfrac{\operatorname{var}(\hat{P})}{P^2}<\infty, \label{eq:BRE}
    \end{equation} the estimator is said to be of bounded relative error. As a result, the number of samples required to achieve a predetermined accuracy with a fixed probability is bounded as the threshold decreases to zero. Hence, if the event is made rarer by decreasing the threshold to zero, the computational cost will not blow up.

    \subsection{Importance Sampling Methods}

    We present four importance sampling methods. The first two of them can be viewed as special cases of what we call selection sampling. 

    Instead of using the nominal (original) distribution, say with the PDF $f_{\rv{X}}$, as in naive MC, importance sampling \cite{glassermanMonteCarlo2004} employs the fact that
    \begin{align*}
        \mathbb{P}[H(\rv{X})\leq\gamma_{\mathrm{th}}]
        &=\int_{H(\dv{x})\leq\gamma_{\mathrm{th}}} \dfrac{f_{\rv{X}}(\dv{x})}{f^*_{\rv{X}}(\dv{x})}f^*_{\rv{X}}(\dv{x})d\dv{x}\\
            &=\mathbb{E}_{f^*_{\rv{X}}(\dv{x})}\left[\mathbf{1}_{\{H(\dv{x})\leq\gamma_{\mathrm{th}}\}}\dfrac{f_{\rv{X}}(\dv{x})}{f^*_{\rv{X}}(\dv{x})}\right],
    \end{align*} and uses the importance sampling probability density function (IS PDF) $f^*_{\rv{X}}$ to reduce the variance\footnote{Note that, for the equation to be valid, it should hold that $\operatorname{supp}\left(\mathbf{1}_{\{H(\dv{x})\leq\gamma_{\mathrm{th}}\}}f_{\rv{X}}(\dv{x})\right)\subseteq \operatorname{supp}\left(\mathbf{1}_{\{H(\dv{x})\leq\gamma_{\mathrm{th}}\}}f^*_{\rv{X}}(\dv{x})\right)$.}. The symbol $\mathbf{1}_{\{.\}}$ denotes the indicator function. Hence, an importance sampling estimator can be written as \begin{equation}
        \hat{P}_N=\dfrac{1}{N} \sum_{n=1}^N \hat{\ell}_n,
    \end{equation}
    where $\hat{\ell}_n$ is the estimator applied to the $n^{\text{th}}$ sample, given by \begin{equation}
        \hat{\ell}=\mathbf{1}_{\{H(\rv{X})\leq\gamma_{\mathrm{th}}\}}\mathcal{L}(\rv{X}),
    \end{equation} with $\mathcal{L}(\rv{X})$ being the likelihood ratio, evaluated at the sample vector $\rv{X}$, \begin{equation*}
        \mathcal{L}(\rv{X})=\dfrac{f_{\rv{X}}(\rv{X})}{f^*_{\rv{X}}(\rv{X})}.
    \end{equation*} Since $\operatorname{var}(\hat{P}_N)=\dfrac{1}{N}\operatorname{var}(\hat{\ell})$, it is sufficient to study the variance of $\hat{\ell}$. Note that CLT practically guarantees the normal distribution of the IS estimator for a sufficiently large number of samples. 

    It can be demonstrated that the optimum IS PDF (in the sense of minimizing the variance) is given by \begin{equation}
        g^*_{\rv{X}}(\dv{x})=\dfrac{f_{\rv{X}}(\dv{x})}{P}\mathbf{1}_{\{H(\dv{x})\leq \gamma_{\mathrm{th}}\}}. \label{eq:opt-density}
    \end{equation} In fact, the variance of the importance sampling estimator with the latter PDF is null. However, we do not have access to it as it depends on the constant $P$ whose calculation is the ultimate goal of the problem. In spite of that, it is still known up to a multiplicative constant, the fact that will be utilized later in the cross-entropy method.

        \subsubsection{Selection Sampling}

            Ben Rached et al. employ in \cite{benrachedSumOrderStatistics2018}, what we shall call ``selection sampling'', in order to estimate the OP of SIMO with the GSC combining scheme and under several fading regimes. As they point out, the idea is to find an event $\mathcal{S}_1$ implied by $\mathcal{S}$, i.e., $\mathcal{S}\subset\mathcal{S}_1$, so that (1) the probability of $\mathcal{S}_1$ is available in closed-form, and (2) $\mathcal{S}_1$ is close to $\mathcal{S}$. We can actually add a third condition: sampling from the distribution truncated to $\mathcal{S}_1$ should be feasible. Let $\ell_1=\mathbb{P}(\mathcal{S}_1)$. The method utilizes Bayes law $$P=\ell_1 \mathbb{P}(\mathcal{S}|\mathcal{S}_1),$$ and approximates the probability in the RHS, i.e., simulating from the truncated distribution over $\mathcal{S}_1$, \begin{equation}
                \hat{\ell}=\ell_1 \mathbf{1}_{\{H(\rv{X})\leq \gamma_{\mathrm{th}}\}}.
            \end{equation}  Viewed as an importance sampling method, its IS PDF is given by \begin{equation}
                f^*_{\rv{X}}(x_1,\ldots,x_M)=\dfrac{f_{\rv{X}}(x_1,\ldots,x_M)}{\ell_1}\mathbf{1}_{\{\dv{x}\in \mathcal{S}_1\}}.
            \end{equation}  

            It can be easily demonstrated that the variance of the estimator is $\ell_1 P-P^2$. Clearly, $\ell_1\geq P$. Expectedly, the ``closer'' $\mathcal{S}_1$ is to $\mathcal{S}$ (in the sense of smaller $\ell_1$), the larger is the amount of variance reduction. 

            \paragraph{Universal Importance Sampling}
    
                Assuming the gains are independent, the authors \cite{benrachedSumOrderStatistics2018} suggest a ``universal'' importance sampling technique given by \begin{equation}
                    \mathcal{S}_1=\left\{\max_k |\rs{h}_{k}|^2 \leq \gamma_{\mathrm{th}}\right\}=\left\{\rs{X}^{(1)} \leq \gamma_{\mathrm{th}}\right\}.
                \end{equation} In our case, we have \begin{equation}
                    \ell_1=\prod_{i=1}^MF_{\chi_2^2(2|\mu_i|^2)}(2\gamma_{\mathrm{th}}),
                \end{equation} where $F_{\chi_{2}^2(2|\mu_i|^2)}$ denotes the CDF of $\chi_{2}^2(2|\mu_i|^2)$.  To sample from the truncated distribution to $\mathcal{S}_1$, the $M$ components of the random vector are sampled independently from the univariate distribution $\frac{1}{2}\chi_2^2(2|\mu_i|^2)$ truncated at $\gamma_{\mathrm{th}}$. The inverse transform method \cite[Sec 3.1.1]{kroese2013handbook} is used. Start by generating uniform random variables $\rs{u}_i\sim\mathcal{U}(0,1)$, and apply the inverse non-central chi-square CDF \begin{equation}
                    \rs{X}_i = \frac{1}{2}F^{-1}_{\chi_2^2(2|\mu_i|^2)}(\kappa_i\rs{u}_i),
                \end{equation} where $\kappa_i=F_{\chi_2^2(2|\mu_i|^2)}(2\gamma_{\mathrm{th}})$. This importance sampling method is called universal because it can be actually applied to any fading regime with independent gains.

                The authors \cite{benrachedSumOrderStatistics2018} prove that, for independent $\rs{X}_i, \; i=1,\ldots,M$, if $$\dfrac{\mathbb{P}(\rs{X}_1\leq\gamma_{\mathrm{th}})}{\mathbb{P}\left(\dfrac{\rs{X}_1}{m}\leq\gamma_{\mathrm{th}}\right)}=\mathcal{O}(1)$$ when $\gamma_{\mathrm{th}}\to 0$, the bounded relative error property holds. We can easily verify that the assumption holds in the case of Rician fading.
    
            \paragraph{Partition Importance Sampling}
    
                For our particular problem setting, we can utilize the fact that the sum of non-central chi-squares is again a non-central chi-square to construct a selection sampling method. To the best of our knowledge, this is a novel method. Note that the selection sampling framework (1) does not guarantee the existence of some event $\mathcal{S}_1$ implied by $\mathcal{S}$, with probability available in closed-form, and simultaneously meaningfully close to the outage probability, nor (2) does it provide an efficient sampler in case of existence of such a set. For the sake of brevity, we will first assume that $M=2m$ and $\mu_i=\mu$ for $i=1,2,\ldots,M$.
                
                Consider the event \begin{equation}
                    \mathcal{S}_2=\left\{ \sum_{k=1}^m \rs{X}_{k} \leq \gamma_{\mathrm{th}} \text{ and } \sum_{k=m+1}^M \rs{X}_{k} \leq \gamma_{\mathrm{th}}\right\},
                \end{equation} with probability $\ell_2=\mathbb{P}(\mathcal{S}_2)$. So instead of just forcing every power gain to be less than the threshold, two non-overlapping sums are. Clearly, $\mathcal{S}_2$ implies $\mathcal{S}_1$, the fact that reduces the variance further. Hence, the bounded relative error result still holds. Due to independence and equivariance, we can write $\ell_2$ in closed form \begin{equation}
                    \ell_2=\left[F_{\chi_{2m}^2(2m|\mu|^2)}(2\gamma_{\mathrm{th}})\right]^2.
                \end{equation}
                
                To sample from the truncated distribution to $\mathcal{S}_2$, which is not a standard one, an acceptance-rejection algorithm \cite[Sec 3.1.5]{kroese2013handbook} is used. Given the event $\mathcal{S}_2$ occurs, the set of the first $m$ random variables is independent from the second. Hence the sets can be sampled independently.
                
                Given $\mathcal{S}_2$ occurs, the first $m$ random variables follow the truncated distribution with the joint PDF \begin{equation}
                    \begin{aligned}
                        &f(x_1,\ldots,x_m):=f_{\rs{X}_1,\ldots,\rs{X}_m|\sum_{i=1}^m\rs{X}_i\leq \gamma_{\mathrm{th}}}(x_1,\ldots,x_m)\\
                        &=\dfrac{e^{-m|\mu|^2}}{F_{\chi_{2m}^2(2m|\mu|^2)}(2\gamma_{\mathrm{th}})}\displaystyle\prod_{i=1}^m\exp(- x_i)I_0\left(2|\mu|\sqrt{x_i}\right), \text{ for} \\ &\sum_{i=1}^m x_i\leq \gamma_{\mathrm{th}}, \;  \text{and}\; x_i \geq 0, \text{ for } i=1,\ldots,m. \label{eq:S2-m-PDF}
                    \end{aligned} 
                \end{equation} The proposal distribution is the uniform distributions over the simplices $\{\sum_{i=1}^m x_i\leq \gamma_{\mathrm{th}},\; x_i\geq 0\}$. The proposal density is given by \begin{equation}
                    \begin{aligned}
                        &g(x_1,\ldots,x_m)=\dfrac{m!}{\gamma_{\mathrm{th}}^m}, \text{ for} \\ &\sum_{i=1}^m x_i\leq \gamma_{\mathrm{th}}, \;  \text{and}\; x_i \geq 0, \text{ for } i=1,\ldots,m.
                    \end{aligned}
                \end{equation} Bounding the ratio of the densities is closely related to the problem of finding the mode of the non-central chi-square distribution. We consider two cases based on the size of the mean.
                
                \underline{Case 1:} If $|\mu|\leq 1$, the mode of the one-dimensional distribution is zero \cite{aubelUnimodalityNonCentral2000}, and we have the following bound: \begin{equation}
                    \dfrac{f(x_1,\ldots,x_m)}{g(x_1,\ldots,x_m)}\leq\dfrac{ \gamma_{\mathrm{th}}^m\exp(-m|\mu|^2)}{m! F_{\chi_{2m}^2(2m|\mu|^2)}(2\gamma_{\mathrm{th}})}=:M_\ell.
                \end{equation}  
                
                \underline{Case 2:} If $|\mu|> 1$, the mode of the one-dimensional distribution is greater than zero. In general we need to bound the maximum of the PDF. To that end, we provide the following two propositions to characterize some properties of the non-central chi-square PDF.

                \begin{proposition} \label{prop:mode-bound}
                    Consider the distribution $\chi_2^2(\lambda)$, with $\lambda>2$. Its mode $Z_\lambda$ satisfies \begin{equation}
                        \lambda-2\leq Z_\lambda \leq \lambda - 1 \label{eq:mode-bound}
                    \end{equation}
                \end{proposition}
                
                \begin{proof}
                    See Appendix \ref{sec:proof-append-prop-1}.
                \end{proof}

                \begin{proposition} \label{prop:pdf-bound}
                    Consider the distribution $\chi_2^2(\lambda)$, with $\lambda>2$. Let $x_0$ be any number in the interval $(\lambda-2,\lambda-1)$ and $Z_\lambda$ be the distribution's mode. Then its PDF $f_{\chi_2^2(\lambda)}$ satisfies \begin{equation}
                        \dfrac{f_{\chi_2^2(\lambda)}(Z_\lambda)}{f_{\chi_2^2(\lambda)}(x_0)}\leq  \max\left\{C_1(\lambda,x_0),C_2(\lambda,x_0)\right\} \label{eq:PDF-bound}
                    \end{equation} 
                    where \begin{equation}
                        C_1(\lambda,x_0)=\exp\left[(\lambda-1-x_0)\left(-\dfrac{1}{2}+\dfrac{\sqrt{1+4\lambda x_0}-1}{4x_0}\right)\right],\label{eq:C1}\end{equation} and \begin{equation}
                        C_2(\lambda,x_0)=\exp\left[(x_0-\lambda+2)\left(\dfrac{1}{2}-\dfrac{\sqrt{1+\lambda x_0}-1}{2x_0}\right)\right].\label{eq:C2}
                    \end{equation}
                \end{proposition}

                \begin{proof}
                    See Appendix \ref{sec:proof-append-prop-2}.
                \end{proof}
                
                If the threshold $\gamma_{\mathrm{th}}$ is sufficiently small, i.e., $2\gamma_{\mathrm{th}}$ is less than the mode, the maximum over the interval $[0,2\gamma_{\mathrm{th}}]$ is attained at $2\gamma_{\mathrm{th}}$. Hence we 
                consider two subcases. Firstly, if $\gamma_{\mathrm{th}}$ is less than or equal to $|\mu|^2-1$, hence less than or equal to the mode, as proved in \eqref{eq:mode-bound}, we have \begin{equation}
                    \dfrac{f(x_1,\ldots,x_m)}{g(x_1,\ldots,x_m)}\leq\dfrac{\left[2\gamma_{\mathrm{th}}f_{\chi_2^2(2|\mu|^2)}(2\gamma_{\mathrm{th}})\right]^m}{m!F_{\chi_{2m}^2(2m|\mu|^2)}(2\gamma_{\mathrm{th}})}=:M_\ell. \label{eq:Ml-large-mu-small-gamma}
                \end{equation} Otherwise, if $\gamma_{\mathrm{th}}>|\mu|^2-1$, we use \eqref{eq:PDF-bound}. We chose $x_0=\lambda-2+3/(2\lambda)$. Thus we have the following bound \begin{equation}
                    \dfrac{f(x_1,\ldots,x_m)}{g(x_1,\ldots,x_m)}\leq\dfrac{\left[2\gamma_{\mathrm{th}}Cf_{\chi_2^2(2|\mu|^2)}(A_\mu)\right]^m}{m!F_{\chi_{2m}^2(2m|\mu|^2)}(2\gamma_{\mathrm{th}})}=:M_\ell,
                \end{equation} where $$A_\mu=2|\mu|^2-2+\frac{3}{4|\mu|^2},$$ and \begin{equation}
                    C=\max\{C_1(2|\mu|^2,A_\mu),C_2(2|\mu|^2,A_\mu)\}
                \end{equation} 
                
                The acceptance-rejection algorithm generates a random vector $\rv{U}$ from $g$ and a random variable $\rs{V}\sim\mathcal{U}(0,1)$, and accepts the vector $\rv{U}$ as $[\rs{x}_1,\ldots,\rs{x}_m]$ from $f$ whenever \[\rs{V}\leq \dfrac{f(\rv{U})}{M_\ell\; g(\rv{U})}.\] Otherwise, a new couple $(\rs{V},\rv{U})$ is generated and the process is repeated. The number of trials till acceptance follows a geometric distribution of parameter $\frac{1}{M_\ell}$, so the expected number of trials is $M_\ell$. This allows us to store the expected number of needed random vectors and numbers in a buffer before running through the loop. This buffering significantly speeds MATLAB implementation.

                The second set of random variables is generated similarly. The two arrays are then concatenated to generate a random vector $\rv{X}$ from the original distribution truncated to $\mathcal{S}_2$.
                
                Unfortunately, for a fixed $\gamma_{\mathrm{th}}$, as $|\mu|$ increases, the bound $M_\ell$ increases without limit and the acceptance-rejection method deteriorates. This is shown in the following proposition. 

                \begin{proposition} \label{prop:Mell-explosion}
                    Fix $M$, $m$, and $\gamma_{\mathrm{th}}$. As $|\mu|$ goes to infinity, $\ln(M_\ell)=\mathcal{O}(|\mu|)$.
                \end{proposition}
                \begin{proof}
                    See Appendix \ref{proof-append-mean}.
                \end{proof}

                This growth reflects an increasing mismatch between the uniform distribution over the simplex and the nominal distribution as the line-of-sight component increases. This opens the door to the consideration of other proposal distributions in the future.

                The whole acceptance-rejection method is summarized in Algorithm \ref{alg:cap}.
                \begin{algorithm}
                    \caption{Generation of a sample vector $[\rs{X}_1,\ldots,\rs{X}_m]$ from \eqref{eq:S2-m-PDF}}
                    \label{alg:cap}
                    \begin{algorithmic}[1]
                        \REQUIRE $\mu$, $m$, $\gamma_{\mathrm{th}}$
                        \ENSURE $[\rs{X}_1,\ldots,\rs{X}_m]$
                        \IF{$|\mu|\leq 1$}
                            \STATE $M_\ell \gets \dfrac{ \gamma_{\mathrm{th}}^m\exp(-m|\mu|^2)}{m! F_{\chi_{2m}^2(2m|\mu|^2)}(2\gamma_{\mathrm{th}})}$
                        \ELSE 
                            \IF{$\gamma_{\mathrm{th}}\leq |\mu|^2-1$}
                                \STATE $M_\ell \gets \dfrac{\left[2\gamma_{\mathrm{th}} f_{\chi_2^2(2|\mu|^2)}(2\gamma_{\mathrm{th}})\right]^m}{m!F_{\chi_{2m}^2(2m|\mu|^2)}(2\gamma_{\mathrm{th}})}$
                            \ELSE
                                \STATE $A_\mu \gets 2|\mu|^2-2+\frac{3}{4|\mu|^2}$
                                \STATE $C \gets \max\{C_1(2|\mu|^2,A_\mu),C_2(2|\mu|^2,A_\mu)\}$, $C_1$ and $C_2$ are given by \eqref{eq:C1} and \eqref{eq:C2}.
                                \STATE $M_\ell \gets \dfrac{\left[2\gamma_{\mathrm{th}}Cf_{\chi_2^2(2|\mu|^2)}(A_\mu)\right]^m}{m!F_{\chi_{2m}^2(2m|\mu|^2)}(2\gamma_{\mathrm{th}})}$
                            \ENDIF
                        \ENDIF
                        \REPEAT
                            \STATE Generate $\rv{U}$ from the uniform distribution over $\left\{u_i \geq 0, \sum_{i=1}^m u_i \leq 1\right\}$ using \cite[Alg 3.23]{kroese2013handbook}.
                            \STATE $\rv{U}\gets \gamma_{\mathrm{th}}\rv{U}$. 
                            \STATE Generate $\rs{V}$ from $\mathcal{U}(0,1)$.
                        \UNTIL{$\rs{V}\leq \dfrac{f(\rv{U})}{M_\ell g(\rv{U})}$}
                        \STATE $[\rs{X}_1,\ldots,\rs{X}_m] \gets \rv{U}$
                        \RETURN $[\rs{X}_1,\ldots,\rs{X}_m]$
                    \end{algorithmic}
                \end{algorithm}

                \begin{remark}[General $m$ and $\dv{\mu}$]
                    In general, for any $m<M$ and possibly distinct $\mu_i,\;i=1,\ldots,M$, consider the integer division of $M$ by $m$, writing $M=mq+r,\;q,r\in\mathcal{N}\;,0\leq r < m$. We can employ the event given in \eqref{eq:PIS-general}.
                    \begin{figure*}[!t]
                        \begin{equation}
                            \mathcal{S}_2=\left\{\sum_{k=1+tm}^{(t+1)m}|\rs{h}_{k}|^2\leq \gamma_{\mathrm{th}},\; t=0,\ldots,q-1, \text{ and } \sum_{k=M-r+1}^M|\rs{h}_{k}|^2\leq \gamma_{\mathrm{th}}\right\} \label{eq:PIS-general}
                        \end{equation}
                        \hrulefill
                    \end{figure*} Each one of the $(q+1)$ (or possibly $q$) vectors is generated independently. After that, we proceed similarly. Note that the probability $\ell_2$ can be written as \begin{equation}
                        \ell_2=F_{\chi_{2r}^2(2\delta_q)}(2\gamma_{\mathrm{th}})\prod_{t=0}^{q-1} F_{\chi_{2m}^2(2\delta_t)}(2\gamma_{\mathrm{th}}),
                    \end{equation} where \[\delta_t=\sum_{i=mt+1}^{mt+m}|\mu_i|^2,\] for $t=0,\ldots,q-1$, and \[\delta_q=\sum_{i=mq+1}^{M}|\mu_i|^2.\]
                \end{remark}

                \begin{remark}[UIS vs PIS]
                    To quantify the variance reduction gain in selection sampling when UIS is replaced with PIS, we can demonstrate that, for fixed $M$, $m$, and $\|\dv{\mu}\|$, we have \begin{equation}
                    \frac{\ell_2}{\ell_1} \underset{\gamma_{\mathrm{th}}\to 0^+}{\sim} \frac{1}{r!(m!)^q}
                \end{equation} The proof is quite simple and follows from the fact that \[F_{\chi_{2M}^2(2\|\dv{\mu}\|^2)}(2\gamma_{\mathrm{th}}) \underset{\gamma_{\mathrm{th}}\to 0^+}{\sim} \dfrac{e^{-\|\dv{\mu}\|^2}\gamma_{\mathrm{th}}^M}{\Gamma(M+1)}.\]
                \end{remark}

                The following proposition specifies the optimum partition in the case of non-identical means, i.e., the partition that results in the largest variance reduction.                
                
                \begin{proposition} \label{prop:opt-ell2}
                    Assume, without loss of generality, that the means are sorted as \[|\mu_1|\leq |\mu_2| \leq \ldots \leq |\mu_M|,\] then the optimal partition, i.e., the partition that minimizes $\ell_2$, and hence the corresponding variance, is given by \eqref{eq:PIS-opt}.
                \end{proposition}
                \begin{figure*}[!t]
                    {\begin{equation}
                        \mathcal{S}_2=\left\{ \sum_{1}^r|\rs{h}_{k}|^2\leq \gamma_{\mathrm{th}}  \text{ and } \sum_{k=r+tm+1}^{r+(t+1)m}|\rs{h}_{k}|^2\leq \gamma_{\mathrm{th}},\; t=0,\ldots,q-1,   \right\} \label{eq:PIS-opt}
                    \end{equation}}
                    \hrulefill
                \end{figure*}
                \begin{proof}
                    See Appendix \ref{sec:proof-opt-ell2}.
                \end{proof}

        \subsubsection{Exponential Twisting}

            Ben Rached et al. utilize approximate exponential twisting techniques in \cite{benrachedEfficientImportanceSampling2021} to evaluate the rare event of a sum of nonnegative independent and identically distributed (i.i.d.) random variables being less than a threshold. 
            
            The authors modify the original exponential twisting method as presented in Ridder and Rubinstein \cite{ridderMinimumCrossEntropy2007}. The latter authors optimize the Kullback–Leibler (KL) distance between the original PDF and another one from a family of IS PDFs under which the rare event is no longer rare. In particular, the expected value of the sum under the proposal distribution is the threshold itself. Denoting by $f_{\rv{X}}^*(\dv{x})$ the IS PDF, the analytical solution is given by \[f_{\rv{X}}^*(\dv{x})=\frac{f_{\rv{X}}(\dv{x}) \exp \left(\theta^* \sum_{i=1}^M x_i\right)}{\mathbb{E}_{f_{\rv{X}}}\left[\exp \left(\theta^* \sum_{i=1}^M \rs{X}_i\right)\right]},\] where $\theta^*$ solves the equation \[\frac{\mathbb{E}_{f_{\rv{X}}}\left[\sum_{i=1}^M \rs{X}_i \exp \left(\theta^* \sum_{i=1}^M \rs{X}_i\right)\right]}{\mathbb{E}_{f_{\rv{X}}}\left[\exp \left(\theta^* \sum_{i=1}^M \rs{X}_i\right)\right]}=\gamma_{\mathrm{th}}.\] Utilizing the independence and identical distribution, it can be shown that \(f_{\rv{X}}^*(\dv{x})=\prod_{i=1}^M f_{\rs{X}}^*\left(x_i\right)\), where \begin{equation}
                f_{\rs{X}}^*(x)=\frac{f_{\rs{X}}(x) \exp \left(\theta^* x\right)}{M_{\rs{X}}\left(\theta^*\right)},
            \end{equation} and the root $\theta^*$ solves the equation \begin{equation}
                \frac{M_{\rs{X}}^{\prime}\left(\theta^*\right)}{M_{\rs{X}}\left(\theta^*\right)}=\frac{\gamma_{\mathrm{th}}}{M}. \label{eq:num-root}
            \end{equation}
            
            For probability densities $f_{\rs{X}}(x)$ that converge to a non-zero value (as our non-central chi-square) as $x$ goes to zero: \begin{equation*}
                \lim_{x\to 0}f_{\rs{X}}(x)=b\in \mathbb{R}_{>0},
            \end{equation*} Ben Rached et al. \cite{benrachedEfficientImportanceSampling2021} suggest the following approximation to $f^*_{\rs{X}}(x)$: \begin{equation}
                \tilde{f}_{\rs{X}}(x)=\dfrac{M}{\gamma_{\mathrm{th}}}\exp\left(-\dfrac{M}{\gamma_{\mathrm{th}}}x\right),
            \end{equation} i.e., an exponential distribution with rate $\frac{M}{\gamma_{\mathrm{th}}}$. The approximation circumvents the need to numerically solve the non-linear equation \eqref{eq:num-root}. Moreover, in the case of other fading regimes, the MGF may be unavailable in closed-form, so the latter approximation overcomes that obstacle.
            
            What is different in our case is that we do not consider the full sum, but partial sums of order statistics. However, the method may still work. A heuristic justification may be as follows: if the sum (MRC) falling under the threshold is no longer rare, then the smaller partial sum (GSC/MRC) falling under the same threshold must also be non-rare. 
            
            In our particular case of Rician fading, and for a sample vector $\rv{X}$, the estimator is given by \begin{equation}
                \begin{aligned}
                \hat{\ell}_{\mathrm{ET}}&=\mathcal{L}_{\mathrm{ET}}(\rv{X})\mathlarger{\mathbf{1}}_{\left\{H(\rv{X})\leq \gamma_{\mathrm{th}}\right\}},
                \end{aligned}
            \end{equation} where \begin{equation*}
                \begin{aligned}
                    \mathcal{L}_{\mathrm{ET}}(\rv{X})&=\dfrac{\gamma_{\mathrm{th}}^M e^{-\|\dv{\mu}\|^2}}{M^M}\\
                    &\times\exp\left(\dfrac{M-\gamma_{\mathrm{th}}}{\gamma_{\mathrm{th}}}\sum_{i=1}^M \rs{X}_i\right)\prod_{i=1}^M I_0\left(2|\mu_i|\sqrt{\rs{X}_i}\right).
                \end{aligned}
            \end{equation*}  The efficiency of this estimator is given by the following proposition.
            
            \begin{proposition} \label{prop:ET-BRE}
                Fix $M$, $m$, and $\dv{\mu}$. The approximate exponential twisting method is of bounded relative error, i.e., \begin{equation}
                    \limsup_{\gamma_{\mathrm{th}}\to 0} \dfrac{\operatorname{var}(\hat{\ell}_{\mathrm{ET}})}{P^2} < \infty.
                \end{equation}
            \end{proposition}

            \begin{proof}
                See Appendix \ref{sec:proof-append-1}.
            \end{proof}

            Again, this property demonstrates that the relative error is asymptotically bounded as the threshold approaches zero, the fact which assures that the number of samples needed to achieve a preset accuracy with a fixed probability does not grow with the decrease in probability.

        \subsubsection{Cross-Entropy}

            For this method, we consider i.i.d. gains, i.e., all means are assumed to be equal: \[\mu_i=\mu, \;i=1,\ldots,M.\] Cross-entropy variance reduction method is based on minimizing the KL distance between the optimum IS PDF (as in \eqref{eq:opt-density}) and the proposal one \cite{deBoerTutorialCrossEntropy2005}. The latter is usually constrained to a family of distributions parametrized by a finite-dimensional vector and usually containing the nominal PDF. Denote by $f(\dv{x};\dv{v})$ the chosen family of distributions, and by $\dv{u}$ the parameter corresponding to the nominal distribution, i.e., $f(\dv{x};\dv{u})=f_{\rv{X}}(\dv{x})$. It can be easily shown that the minimization of the mentioned KL-distance is equivalent to the following optimization problem \begin{equation}
                \max_{\dv{v}} \mathbb{E}_{\dv{u}} \ln f(\rv{X};\dv{v}) \mathbf{1}_{\{H(\rv{X})\leq \gamma_{\mathrm{th}}\}}, \label{eq:opt}
            \end{equation} where $\mathbb{E}_{\dv{v}}$ denotes the expectation over the random vector $\rv{X}$ with the PDF $f(\dv{x},\dv{v})$. A naive MC estimator could be used for the optimization, but we would face the same issue of the original problem: the event is rare, so the rejection rate is high. However, if the threshold $\gamma_{\mathrm{th}}$ were sufficiently large, the problem could be solved. Now applying the same trick of importance sampling, \eqref{eq:opt} can be rewritten as \begin{equation}
                \max_{\dv{v}} \mathbb{E}_{\dv{w}} \ln f(\rv{X};\dv{v}) \mathcal{L}(\rv{X};\dv{u},\dv{w}) \mathbf{1}_{\{H(\rv{X})\leq \gamma_{\mathrm{th}}\}}, \label{eq:opt-2}
            \end{equation} where $\mathcal{L}(\rv{X};\dv{u},\dv{w})=\dfrac{f(\rv{X};\dv{u})}{f({\rv{X};\dv{w}})}$. If $\dv{w}$ is properly chosen, the outage event will become no longer rare and the optimization is doable. Again, this is the raison d'\^{e}tre of importance sampling. 

            To overcome the mentioned obstacles, the method utilizes a sequence of optimization problems with decreasing thresholds. The idea is to start from a threshold that will produce a moderately small probability. In particular, $N_0$ samples are generated according to the nominal distribution, and the $\alpha$-th quantile of the random variable $H(\rv{X})$ is estimated. Usually, $\alpha$ is chosen to be $10\%$. The quantile is assigned to the first auxiliary threshold, $\hat{\gamma}_1$. Equation \eqref{eq:opt} is numerically solved using the generated samples and $\hat{\gamma}_1$ in place of $\gamma_{\mathrm{th}}$. Denote by $\dv{v}_1$ the maximizer of the latter equation. After that, a new set of $N_0$ vectors is sampled from $f(\dv{x},\dv{v}_1)$, and the auxiliary threshold $\hat{\gamma}_2$ is assigned the new $\alpha$-th quantile. If it is less than the primary threshold $\gamma_{\mathrm{th}}$, this implies that the outage event is no longer rare and the distribution can be used for importance sampling. Otherwise, optimization \eqref{eq:opt-2} is solved numerically with the expectation over $f(\dv{x},\dv{v}_1)$ and the auxiliary threshold $\hat{\gamma}_2$. The optimization and threshold adjustment are repeated until the auxiliary threshold is less than the target one, in which case, the auxiliary threshold is reset to the nominal one and a final optimization yields the parameter vector $\dv{v}_T$ that will be used for importance sampling.

            In our particular case, we select the family of independent scaled identical non-central chi-squares of two degrees of freedom, where the scale and the non-centrality parameters are unknown: $v_1 \chi_2^2(v_2)$, whose PDF is given by \begin{equation}
                f(\dv{x},\dv{v})=\prod_{i=1}^M \dfrac{e^{-v_2/2}}{2v_1}\exp\left(-\dfrac{x_i}{2v_1}\right)I_0\left(\sqrt{\dfrac{v_2}{v_1}x_i}\right)
            \end{equation} Evidently, the couple $\dv{v}_0=\dv{u}=[\frac{1}{2},2|\mu|^2]$ corresponds to the nominal distribution. The method is illustrated in Algorithm \ref{alg:CE}.

            \begin{algorithm}
                    \caption{Cross Entropy Estimation of P as in \cite{deBoerTutorialCrossEntropy2005}}
                    \label{alg:CE}
                    \begin{algorithmic}[1]
                        \REQUIRE $m$, $M$, $\mu$, $\gamma_{\mathrm{th}}$, $N$, $N_0$, $\alpha$
                        \ENSURE $\hat{P}_N$
                        \STATE Generate $\rv{X}_1, \rv{X}_2, \ldots, \rv{X}_{N_0}$ with independent components and marginal PDFs $f(\dv{x},\dv{v}_0)$, $\dv{v}_0=\dv{u}=[\frac{1}{2},2|\mu|^2]$.
                        \STATE Sort the corresponding partial sums $\rs{H}^{(1)}\leq \rs{H}^{(2)} \leq \ldots \leq \rs{H}^{(N_0)}$.
                        \STATE Compute $\hat{\gamma}_1=\rs{H}^{(\lfloor \alpha N_0 \rfloor)}$.
                        \STATE Set $t=1$.
                        \WHILE{$\hat{\gamma}_t\geq \gamma_{\mathrm{th}}$}
                            \STATE Solve numerically to get $\dv{v}_t$ \begin{equation}
                                \max_{\dv{v}} \sum_{n=1}^{N_0} \ln f(\rv{X}_n;\dv{v}) \mathcal{L}(\rv{X}_n;\dv{v}_0,\dv{v}_{t-1}) \mathbf{1}_{\{H(\rv{X}_n)\leq \hat{\gamma}_t\}} \label{eq:opt-num}
                            \end{equation} 
                            \STATE Generate $\rv{X}_1, \rv{X}_2, \ldots, \rv{X}_{N_0}$ with independent components and marginal PDFs $f(\dv{x},\dv{v}_t)$.
                            \STATE Sort the corresponding partial sums $\rs{H}^{(1)}\leq \rs{H}^{(2)} \leq \ldots \leq \rs{H}^{(N_0)}$.
                            \STATE Compute $\hat{\gamma}_{t+1}=\rs{H}^{(\lfloor \alpha N_0 \rfloor)}$.
                            \STATE Set $t=t+1$.
                        \ENDWHILE
                        \STATE Solve numerically to get $\dv{v}$ \begin{equation}
                                \max_{\dv{v}} \sum_{n=1}^{N_0} \ln f(\rv{X}_n;\dv{v}) \mathcal{L}(\rv{X}_n;\dv{v}_0,\dv{v}_{t-1}) \mathbf{1}_{\{H(\rv{X}_n)\leq \gamma_{\mathrm{th}}\}} \label{eq:opt-final}
                            \end{equation}
                        \STATE Generate $\rv{X}_1, \rv{X}_2, \ldots, \rv{X}_{N}$ with independent components and marginal PDFs $f(\dv{x},\dv{v})$.
                        
                        \STATE Estimate $\hat{P}_N$ using \begin{equation}
                            \hat{P}_N=\dfrac{1}{N} \sum_{n=1}^N \mathcal{L}(\rv{X}_n;\dv{v}_{0},\dv{v}) \mathbf{1}_{\{H(\rv{X}_n)\leq \gamma_{\mathrm{th}}\}}
                        \end{equation}
                        \RETURN $\hat{P}_N$.
                    \end{algorithmic}
                \end{algorithm}

                \begin{remark}
                    Varying the scale and non-centrality parameter of $v_1 \chi_2^2(v_2)$ is equivalent to shifting and scaling the underlying complex Gaussian. However, the fundamental theorem of system simulation \cite[Sec 4.3]{bucklewIntroductionRareEventSimulation2004} implies that the latter IS produces a larger variance than the former.
                \end{remark}


                \begin{remark}
                    Fix $m$, $M$, and $\mu$. Assume that, for some $\epsilon>0$, the algorithm converges for every $\gamma_{\mathrm{th}}\leq \epsilon$, and the optimizer of \eqref{eq:opt-final} satisfies \begin{equation}
                        v_1=\mathcal{O}(\gamma_{\mathrm{th}}) \text{ as } \gamma_{\mathrm{th}}\to0^+,
                    \end{equation} and \begin{equation}
                        v_2=\mathcal{O}(1) \text{ as } \gamma_{\mathrm{th}}\to0^+.
                    \end{equation} Then the cross-entropy method is of bounded relative error, i.e., \begin{equation}
                    \limsup_{\gamma_{\mathrm{th}}\to 0} \dfrac{\operatorname{var}(\hat{\ell}_{\mathrm{CE}})}{P^2} < \infty.
                \end{equation} The proof is very similar to that of ET: bound the modified Bessel function in the denominator from below by one and proceed similarly. Note that the assumed asymptotic regimes are observed later in the numerical results.
                \end{remark}
                

    \subsection{Multilevel Splitting}

        The authors of \cite{benrachedUniversalSplittingEstimator2020} embed the studied rare event in a continuous time Markov process in order to estimate its probability, making use of the approach in \cite{kahnEstimationRandomSampling1951}.

        The authors start with a multivariate gamma process $\rs{G}_i(t),\; t\in[0,1],\; i=1,\ldots,M$. Denote by $F_i^{-1}=F^{-1}_{\rs{X}_i}$ the inverse CDF of the independent random variables $\rs{X}_i$. Another stochastic process is generated from the former via \begin{equation}
            \rs{X}_i(t)=F_i^{-1}\left(1-\exp(-\rs{G}_i(t))\right).
        \end{equation} This sequence satisfies $\rs{X}_i(1)\stackrel{d}{=}\rs{X}_i$. Since the sum of the largest $m$ gains, $H(\rv{X}(t))$, is an increasing function in each component $\rs{X}_i(t)$, it follows that it is also an increasing function of $t$. The multivariate gamma process is simulated at a sequence of levels $t_0=0<t_1<\ldots<t_L=1$. Due to the monotone nature of the function $H$, we have \(\{H(\rv{X}(t_\ell))\leq \gamma_{\mathrm{th}}\}\subset \{H(\rv{X}(t_{\ell-1}))\leq \gamma_{\mathrm{th}}\}\). Then we can write the sought event probability as the product of conditional probabilities as follows \begin{equation}
            \begin{aligned}
                P&=\mathbb{P}[H(\rv{X}(1)\leq \gamma_{\mathrm{th}}]\\
                &=\prod_{\ell=1}^L\mathbb{P}\left[H(\rv{X}(t_{\ell})\leq \gamma_{\mathrm{th}} \;|\; H(\rv{X}(t_{\ell-1})\leq \gamma_{\mathrm{th}}\right].
            \end{aligned}
        \end{equation} The process starts from $0$. At the stage $t_1$, a set of $N$ increments following $\mathrm{Gamma}(t_1,1)$ is generated, and those falling under the threshold $\gamma_{\mathrm{th}}$ \[H(\rv{X}(t_1))\leq \gamma_{\mathrm{th}}\] are collected in the set $\chi_1$. At each stage $\ell$, $N$ samples are drawn via adding a $\mathrm{Gamma}(t_\ell-t_{\ell-1},1)$ random variable to a uniformly randomly chosen sample from the set of previous favourable outcomes, i.e., that satisfy $H(\rv{X}(t_{\ell-1})\leq \gamma_{\mathrm{th}}$. Denoting each set of favourable outcomes by $\mathbf{\chi}_\ell$, the final probability can be written as \begin{equation}
            \hat{P}=\prod_{\ell=1}^L\dfrac{|\mathbf{\chi_\ell}|}{N}.
        \end{equation} 
        
        Note that the authors \cite{benrachedUniversalSplittingEstimator2020} give a pilot algorithm in order to predetermine the levels $t_1,\ldots,t_{L-1}$ so that they can be meaningfully used to calculate the sought probability. It is based on the idea that, at each stage, the outage event should not be rare.

\section{Estimation Methods: Correlated Fading} \label{sec:methods-dependent}

    We now consider the more challenging case of correlated Rician fading. In this section, we present two methods to estimate the outage probability. The first one, approximate Exponential Twisting for Correlated channels (ETC), applies to an arbitrary mean and an arbitrary covariance matrix. The second method, Cross-Entropy for Correlated channels (CEC), applies only in the case of equal means, equal variances, and constant correlation.

    A common theme in this section is that correlation blocks the access we had to the joint distribution of non-central chi-squares. Note that this topic \footnote{In the literature, two equivalent problems are studied: the joint distribution of (non-central) chi-squares, and the joint distribution of their square roots. The latter appears typically in the study of the outage probability of SIMO under Rician fading and with selection combining (SC).} has attracted a number of researchers \cite{royenExpansionsMultivariateChisquare1991,royenCENTRALNONCENTRALMULTIVARIATE1995,beaulieuNovelSimpleRepresentations2011,wiegandSeriesApproximationsRayleigh2019,dharmawansaDiagonalDistributionComplex2009}. No similar closed-form expressions are available for either the PDF or the CDF. As a result, the two selection sampling methods cannot be easily extended to the general case. Regarding MLS \cite{benrachedUniversalSplittingEstimator2020}, we need to start from a vector of independent random variables and apply a pseudo-monotone function to it. The chi-squares are no longer independent, and starting from the Gaussians violates pseudo-monotonicity.

    \subsection{Exponential Twisting}

        As we mentioned before, we can no longer start from the chi-squares. Inspired by the method before, we decorrelate, shift, and scale the Gaussians so that \[\rs{X}_i \sim\mathrm{Exp}(M/\gamma_{\mathrm{th}})\] Hence the proposal distribution is given by \begin{equation}
            \rv{h}\sim \mathcal{CN}\left(\dv{0},\dfrac{\gamma_{\mathrm{th}}}{M}\dmt{I}_M\right)
        \end{equation} The likelihood ratio is thus a function of the complex channel vector, and it is given by \begin{equation}
            \begin{aligned}
                \mathcal{L}_{ETC}(\rv{h})&=\dfrac{\gamma_{\mathrm{th}}^M}{M^M |\dmt{\Sigma}|}\\
                &\times\exp\left(\dfrac{M}{\gamma_{\mathrm{th}}}\rv{h}^H\rv{h}-(\rv{h}-\dv{\mu})^H\dmt{\Sigma}^{-1}(\rv{h}-\dv{\mu})\right)
            \end{aligned}
        \end{equation}

        The following proposition proves that the new ET also has a bounded relative error as $\gamma_{\mathrm{th}}$ goes to zero. 
        \begin{proposition} \label{prop:ETC-BRE}
            Fix $M$, $m$, $\dv{\mu}$, and $\dmt{\Sigma}$. The approximate exponential twisting method is of bounded relative error, i.e., \begin{equation}
                    \limsup_{\gamma_{\mathrm{th}}\to 0} \dfrac{\operatorname{var}(\hat{\ell}_{\mathrm{ETC}})}{P^2} < \infty.
                \end{equation}
        \end{proposition}
        \begin{proof}
            See Appendix \ref{sec:proof-BRE-ETC}.
        \end{proof}
        As we can see, the bounded error property holds for an arbitrary mean and an arbitrary covariance matrix. This method is the most widely applicable in the paper, and, simultaneously, its desirable BRE property does not impose any additional structure.

    \subsection{Cross-Entropy}

        For this method, assume that the means are identical $\mu_i=\mu$, the variances are equal and without loss of generality normalized to unity $\sigma_i^2=1$, and the covariance structure is of constant correlation, i.e., \[\rho_{ij}=\rho,\text{ for }i\neq j.\] Again, we work with the Gaussians instead of the squared norms. The family of proposal distributions we suggest is \[\rv{h}\sim\mathcal{CN}(\nu \mathbf{1}_M,\sigma^2\dmt{\Sigma}_{\zeta})\] where \[\dmt{\Sigma}_{\zeta}=\begin{bmatrix}
            1 & \zeta &\cdots & \zeta        \\
            \zeta & 1 & \ddots & \vdots     \\
            \vdots & \ddots & \ddots &\zeta \\
            \zeta & \cdots & \zeta & 1
        \end{bmatrix}.\] The nominal distribution corresponds to $\nu=\mu$, $\zeta=\rho$, and $\sigma^2=1$. Order the parameters in $\dv{v}=[v_1,v_2,v_3]^T=[\nu,\sigma^2,\zeta]$. Hence the proposal PDF is given by \begin{equation}
            \begin{aligned}
                &f(\dv{h},\dv{v})=\dfrac{1}{\pi^Mv_2^{M}(1-v_3)^{M-1}[1+(M-1)v_3]}\\
                &\times\exp\left(-\dfrac{1}{v_2(1-v_3)}(\dv{h}-v_1\mathbf{1})^H\right.\\
                &\left.\left(\dmt{I}_M-\dfrac{v_3}{1+(M-1)v_3}\dmt{J}_M\right)(\dv{h}-v_1\mathbf{1})\right)
            \end{aligned}
        \end{equation} where $\dmt{J}_M$ is the $M\times M$ matrix of ones. Hereafter, we continue in exactly the same manner as the previous CE method.

\section{Theoretical Comparison: A Tabular Summary} \label{sec:comparison-theoretical}

To summarize the section, we compare the simulation methods according to their applicability and theoretical guarantees. For brevity, the comparison is limited to Table \ref{tab:theo-summary}. As we can see, all methods are guaranteed theoretically to converge to the sought probability except for CE and CEC. In addition, MLS, unlike UIS, PIS, and ET, is not guaranteed to have bounded relative error (BRE) as the threshold goes to zero. 

\begin{table}[ht]
    \centering
    \resizebox{\columnwidth}{!}{\begin{tabular}{|c|c|c|c|c|c|}
        \hline
        Method & Different LOS Power & Different Scattering Powers & Correlation & Convergence & BRE \\ \hline
        UIS & Y & Y & N  & Y & Y \\ \hline
        PIS & Y & N & N & Y & Y \\ \hline
        ET  & Y & N & N & Y & Y \\ \hline
        CE  & N & N & N & N & N \\ \hline
        MLS & Y & Y & N & Y & N \\ \hline
        CEC & N & N & Y\footnotemark & N & N \\ \hline
        ETC & Y & Y & Y & Y & Y \\ \hline
    \end{tabular}}
    \caption{Theoretical Comparison of Simulation Methods}
    \label{tab:theo-summary}
\end{table} \footnotetext{It only works for constant correlation.}

\section{Numerical Comparison} \label{sec:comparison-numerical}

    The estimated relative error defined in \eqref{eq:ERE} is informative of the quality of an estimate. However, as we are interested in efficiency, the time taken is included in the \textit{work normalized relative variance} (WNRV) \cite[Sec. 10.1]{kroese2013handbook}, which is the squared relative error multiplied by the computation time in seconds: \begin{equation}
        \mathrm{WNRV}=\dfrac{\operatorname{var}(\hat{P})}{P^2}\times \mathrm{Time} \;\mathrm{(s)} , \label{eq:WNRV}
    \end{equation} where $P$ is approximated using the very same method, as in \cite{benrachedSumOrderStatistics2018,benrachedUniversalSplittingEstimator2020,benrachedEfficientImportanceSampling2021}, the variance is estimated from the drawn samples \footnote{In the cases of NMC, UIS, and PIS, the variance can be estimated, for $N=1$, using the formulas $P-P^2$, $\ell_1 P-P^2$, and $\ell_2 P - P^2$ respectively. Since the MLS estimator is not an average of i.i.d. random variables, it must be repeated $n$ times to estimate its variance. In our experiments, we took $n=50$.}, and $\mathrm{Time} \;\mathrm{(s)}$ is the computation time in seconds the estimator takes. Theoretically, the WNRV should not significantly change upon an increase in the number of samples, since the time increases and the variance decreases by the same factor. All implementations hereafter were carried out on a standard personal computer using MATLAB. The CE threshold $\alpha$ is set to $0.1$ in all subsequent simulations. The built-in MATLAB optimizer is used, with the nominal distribution parameters as the initial vector. In the optimization phase, $N_0=10^3$ samples are used to estimate the pertinent expectations. Finally, note that the code used to generate Tables II, III, VII, and VIII is available on \href{https://github.com/MahmoudGhazal96/GSC-Rician}{GitHub}. 

    \subsection{Case of Uncorrelated Fading}
    
    Shown in Table \ref{tab:variance} are the results of the all previous MC methods, for the following input: $M=8$, $m=4$, $\mu=0.5$, and $\gamma_{\mathrm{th}}\in\{0.5,1\}$, assuming all channels follow $\mathcal{CN}(\mu,1)$. In addition, a deterministic method, the MGF-inversion based method in \cite{annamalaiAnalysisGeneralized2006}, is used. The paper's recommended hyper-parameters are utilized. It is observed that PIS, ET, and CE are the best contenders for the chosen values, with CE scoring the lowest WNRV. We can clearly see that the naive MC method is, as expected, extremely inefficient for rare events. In addition to that, UIS, which has the generality advantage over PIS, does much worse, gaining a several thousand-fold WNRV. On the other hand, MLS is not comparable, performance-wise, to the best three importance sampling methods (its WNRV is $100$ to $300$ times larger than that of PIS). Finally, the MGF-inversion based method validates the stochastic estimates.

    \begin{table*}[!t]
        \centering
        \caption{Variance Summary, All Methods, $M=8$, $m=4$, $\mu=0.5$}
        \label{tab:variance}
        \resizebox{\textwidth}{!}{\begin{tabular}{|c||c|c|c|c|c||c|c|c|c|c|}
            \hline
             & \multicolumn{5}{c||}{$\gamma_{\mathrm{th}} = 1$} & \multicolumn{5}{c|}{$\gamma_{\mathrm{th}}=0.5$} \\ \hline
             Method & N & $\hat{P}_N$  & Time (s) & RE (\%) & WNRV &N & $\hat{P}_N$  & Time (s) & RE (\%) & WNRV \\ \hline
             NMC & $10^9$ &$9.246\times 10^{-6}$ &$1498$ &$1.0$ &$1.6\times10^{-1}$    &$10^{10}$ &$5.57\times10^{-8}$ &$12685$ &$4.2$ &$2.3\times10^{+1}$ \\ \hline
             UIS & $5\times10^6$ &$9.358\times10^{-6}$ &$1074$ &$1.3$ &$1.8\times10^{-1}$    &$5\times10^6$ &$5.637\times10^{-8}$ &$1098$ &$2.1$ &$4.8\times10^{-1}$ \\ \hline
             PIS & $10^6$ &$9.214\times 10^{-6}$ &$8.7$ &$0.26$ &$5.8\times 10^{-5}$    &$10^6$ &$5.577\times 10^{-8}$ &$4.9$ &$0.28$ &$3.9\times 10^{-5}$ \\ \hline 
             MLS & $1.6\times10^5$ &$9.200\times 10^{-6}$ &$505$ &$0.40$ &$8.1\times10^{-3}$    &$1.6\times10^5$ &$5.545\times10^{-8}$ &$695$ &$0.42$ &$1.2\times10^{-2}$ \\ \hline
             ET  & $10^6$ &$9.205\times10^{-6}$ &$3.4$ &$0.25$ &$2.1\times10^{-5}$    &$10^6$ &$5.549\times 10^{-8}$ &$3.8$ &$0.16$ &$2.8\times10^{-5}$ \\ \hline
             CE  & $10^6$ &$9.227\times10^{-6}$ &$4.4$ &$0.15$ &$9.6\times10^{-6}$    &$10^6$ &$5.555\times 10^{-8}$ &$3.5$ &$0.16$ &$9.2\times10^{-6}$ \\ \hline
             ETC & $10^6$ & $9.235\times 10^{-6}$ & $3.2$ & $0.35$ & $4.0\times 10^{-5}$ & $10^6$ & $5.560\times 10^{-8}$ & $3.2$ & $0.32$ & $3.4\times10^{-5}$\\ \hline
             CEC & $10^6$ & $9.190\times10^{-6}$ & $1.2$ & $0.21$ & $5.4\times 10^{-6}$ & $10^6$ & $5.550\times 10^{-8}$ & $1.2$ & $0.2$ & $6.0\times10{-6}$\\ \hline
             {Annamalai \cite{annamalaiAnalysisGeneralized2006}} & {-} & {$9.223\times 10^{-6}$}  & {$251$}  & - & {-} & {-} & {$5.561\times10^{-8}$} & {$231$} & {-} & {-} \\ \hline 
             \end{tabular}}
    \end{table*}
    
    Furthermore, rarer events are tested for the best contenders. The event can be made rarer by either decreasing $\gamma_{\mathrm{th}}$ or increasing $\mu$. In Table \ref{tab:variance-rarer}, the mean $\mu$ is fixed at $\mu=0.5$ and the threshold is varied from $0.4$ to $0.2$. We observe that the methods are still close. Indeed, the WNRV of the best method (CE) is $4$ to $6$ times smaller than that of PIS, and ET lies in between. Note that the three methods still do well even for lower thresholds, but the probabilities become unrealistically rare.
    
    With $M=8$ and $m=2$ as in Table \ref{tab:variance-rarer-3}, the results show that PIS and CE perform remarkably well whereas ET deteriorates. Its WNRV increases by two orders of magnitude compared to the $m=4$ case, whereas PIS and CE still do well. This is probably due to the fact that ET was `optimized' for full sums, so the method works better for larger subsets, i.e., when $m$ is closer to $M$. Although the outage event is no longer rare under the importance sampling distribution in the case of smaller subsets, the variance of the whole estimator (with the likelihood ratio) becomes larger. For example, in the case of $m=2$, $M=8$, $\gamma=0.1$, and $\mu=0.5$, the outage event occurs $96\%$ of the time, the total time taken for an estimate with $10^6$ samples is less than $4$ seconds, but the WNRV is in the order of $10^{-2}$ due to the large variance. As put in \cite[Sec 4.2]{bucklewIntroductionRareEventSimulation2004}, the IS biased distribution may hit more unlikely regions in the subset of interest.

    \begin{table*}[ht!]
        \centering
        \caption{Variance Summary with Rarer Events, $M=8$, $m=4$, $\mu=0.5$, $N=10^6$}
        \label{tab:variance-rarer}
        \resizebox{\textwidth}{!}{\begin{tabular}{|c||c|c|c||c|c|c||c|c|c|}
            \hline
             &\multicolumn{3}{c||}{PIS} &\multicolumn{3}{c||}{ET} &\multicolumn{3}{c|}{CE} \\ \hline
             $\gamma_{\mathrm{th}}$ &$\hat{P}_N$ &RE (\%) &WNRV &$\hat{P}_N$ &RE (\%) &WNRV &$\hat{P}_N$ &RE (\%) &WNRV \\ \hline
             $0.4$ &$1.016\times10^{-8}$ &$0.29$ &$4.9\times10^{-5}$ &$1.016\times10^{-8}$ &$0.28$ &$2.7\times10^{-5}$ &$1.019\times10^{-8}$ &$0.16$ &$9.1\times10^{-6}$\\ \hline
             $0.3$ &$1.115\times10^{-9}$ &$0.29$ &$4.2\times10^{-5}$ &$1.109\times10^{-9}$ &$0.28$ &$2.8\times10^{-5}$ &$1.115\times10^{-9}$ &$0.16$ &$8.6\times10^{-6}$ \\ \hline
             $0.2$ &$4.761\times10^{-11}$ &$0.30$ &$4.0\times10^{-5}$ &$4.723\times10^{-11}$ &$0.29$ &$2.8\times10^{-5}$ &$4.739\times10^{-11}$ &$0.18$ &$1.1\times10^{-5}$ \\ \hline
        \end{tabular}}
    \end{table*} 

    \begin{table*}[!t]
        \centering
        \caption{Variance Summary with Rarer Events, $M=8$, $m=2$, $\mu=0.5$, $N=10^6$}
        \label{tab:variance-rarer-3}
        \resizebox{\textwidth}{!}{\begin{tabular}{|c||c|c|c||c|c|c||c|c|c|}
            \hline
             &\multicolumn{3}{c||}{PIS} &\multicolumn{3}{c||}{ET}&\multicolumn{3}{c|}{CE}  \\ \hline
             $\gamma_{\mathrm{th}}$ &$\hat{P}_N$ &RE (\%) &WNRV &$\hat{P}_N$ &RE (\%) &WNRV &$\hat{P}_N$ &RE (\%) &WNRV \\ \hline
             $1$ &$2.339\times10^{-4}$ &$0.20$ &$5.1\times10^{-5}$ &$2.395\times10^{-4}$ &$3.3$ &$4.6\times10^{-3}$ &$2.335\times10^{-4}$ &$0.16$ &$7.6\times10^{-6}$ \\ \hline
             $0.75$ &$3.394\times10^{-5}$ &$0.22$ &$4.9\times10^{-5}$ &$3.485\times10^{-5}$ &$3.8$ &$5.7\times10^{-3}$ &$3.385\times10^{-5}$ &$0.18$ &$9.2\times10^{-6}$\\ \hline
             $0.5$ &$1.928\times10^{-6}$ &$0.23$ &$4.7\times10^{-5}$ &$1.985\times10^{-6}$ &$4.3$ &$7.5\times10^{-3}$ &$1.925\times10^{-6}$ &$0.18$ &$9.4\times10^{-6}$\\ \hline
             $0.25$ &$1.099\times10^{-8}$ &$0.25$ &$4.0\times10^{-5}$ &$1.136\times10^{-8}$ &$4.8$ &$9.7\times10^{-3}$ &$1.099\times10^{-8}$ &$0.20$ &$1.1\times10^{-5}$ \\ \hline
             $0.1$ &$9.064\times 10^{-12}$ &$0.25$ &$4.2\times 10^{-5}$   &$9.386\times 10^{-12}$ &$5.2$ &$1.1\times 10^{-2}$ &$9.038 \times 10^{-12}$ &$0.21$ &$1.4\times10^{-5}$\\ \hline            
        \end{tabular}}
    \end{table*}

    The numerical experiment is repeated, fixing $\gamma_{\mathrm{th}}$, and varying the parameter $\mu$. The results are shown in Table \ref{tab:variance-means-gamma1}. Here we can see that PIS becomes worse as the mean increases. This can be justified by the bound $M_\ell$ that increases, for $\gamma_{\mathrm{th}}=17$, from $6.15$ at $\mu=2.3$, to $313.6$ at $\mu=3$. As we have mentioned before, the bound $M_\ell$ is also the expected number of trials to generate one sample in the corresponding acceptance-rejection method, hence its increase will blow up the computation time. As evident from the table, ET can also deteriorate when the mean increases.

    \begin{table*}[!t]
        \centering
        \caption{Variance Summary with Different Means, $M=8$, $m=4$, $N=10^6$}
        \label{tab:variance-means-gamma1}
        \resizebox{\textwidth}{!}{\begin{tabular}{|c|c||c|c|c||c|c|c||c|c|c|}
            \hline
             \multicolumn{2}{|c||}{} &\multicolumn{3}{c||}{PIS} &\multicolumn{3}{c||}{ET} &\multicolumn{3}{c|}{CE} \\ \hline
             $\gamma_{\mathrm{th}}$ &$\mu$  &$\hat{P}_N$ &RE (\%) &WNRV  &$\hat{P}_N$ &RE (\%) &WNRV  &$\hat{P}_N$ &RE (\%) &WNRV \\ \hline\hline
             \multirow{3}{*}{$1$} &$1$  &$5.157\times 10^{-8}$ &$0.30$ &$5.7\times 10^{-5}$  &$5.139\times 10^{-8}$ &$0.29$ &$3.0\times 10^{-5}$  &$5.156 \times 10^{-8}$ &$0.17$ &$9.3\times 10^{-6}$\\ \cline{2-11}
             &$1.25$  & $1.022\times10^{-9}$ &$0.33$ & $1.2\times10^{-4}$ &$1.022\times10^{-9}$ &$0.33$ &$4.1\times10^{-5}$ &$1.027\times10^{-9}$ &$0.16$ &$8.0\times10^{-6}$\\ \cline{2-11}
             &$1.5$ & $8.299\times10^{-12}$ & $0.37$ & $4.4\times10^{-4}$   &$8.299\times10^{-12}$ &$0.38$ &$5.3\times 10^{-5}$ &$8.349\times10^{-12}$ &$0.20$ &$1.4\times10^{-5}$ \\ \hline\hline
             \multirow{3}{*}{$17$} &$2.3$ & $9.004\times10^{-4}$ & $0.34$ & $8.6\times10^{-4}$ &$9.028\times10^{-4}$ &$3.2$ &$6.1\times10^{-3}$ &$9.004\times10^{-4}$ &$0.17$ &$1.5\times10^{-5}$\\ \cline{2-11}
             &$2.5$ & $6.646\times10^{-5}$ & $0.44$ & $2.9\times10^{-3}$ &$6.834\times10^{-5}$ &$5.6$ &$1.8\times10^{-2}$ &$6.692\times10^{-5}$ &$0.21$ &$2.0\times10^{-5}$ \\ \cline{2-11}
             &$3$ & $1.061\times10^{-8}$ & $0.78$ & $3.5\times10^{-1}$ &$1.200\times10^{-8}$ &$16$ &$1.5\times10^{-1}$ &$1.067\times10^{-8}$ &$0.30$ &$3.4\times10^{-5}$\\ \hline
        \end{tabular}}
    \end{table*}

    The WNRV metric incorporates sampling time, so another possible metric to use is the \textit{squared coefficient of variation} (SCV) which is defined by \begin{equation}
        \mathrm{SCV}=\dfrac{\operatorname{var}(\hat{\ell})}{P^2},
    \end{equation} where the probability and the variance of the estimate are estimated as we mentioned before for \eqref{eq:WNRV}. Compared to the estimated relative error (ERE), it does not take into consideration the number of samples $N$. ERE will converge to zero for any consistent estimator as $N$ goes to infinity. On the other hand, SCV may be viewed as an answer to the following question: ignoring the efficiency of the required sampling method, which estimation method is better? Note that, using CLT, it can be demonstrated that the number of samples required to reach a preset accuracy with a fixed probability is proportional to the SCV.

    To study further our numerical experiments, we plot the SCV versus $\gamma_{\mathrm{th}}$ fixing $\mu$ and vice versa in the following graphs. In Fig. \ref{fig:SCV_mu_half_m_4}, we can see that for $M=8$, $m=4$, and $\mu=0.5$, the SCV plots reinforce the relative efficiency comparison shown in Table \ref{tab:variance-rarer} using WNRV. We have included much lower thresholds with extremely rare probabilities just to showcase the bounded relative error property of PIS and ET. Despite having no theoretical guarantees, CE seems to do better in both WNRV and SCV fronts. It is noteworthy that, for large means, we can see in Fig. \ref{fig:SCV_gamma_1} and Fig. \ref{fig:SCV_gamma_17} that PIS has a better SCV than ET, especially in the case of $\gamma_{\mathrm{th}}=17$. This may suggest that the reason of the inefficiency of PIS in the case of large means is mainly due to inefficiency of the acceptance-rejection algorithm. In all graphs, we can see that CE always has the lowest SCV. 

    \begin{figure}
        \centering
        \includegraphics[width=\linewidth]{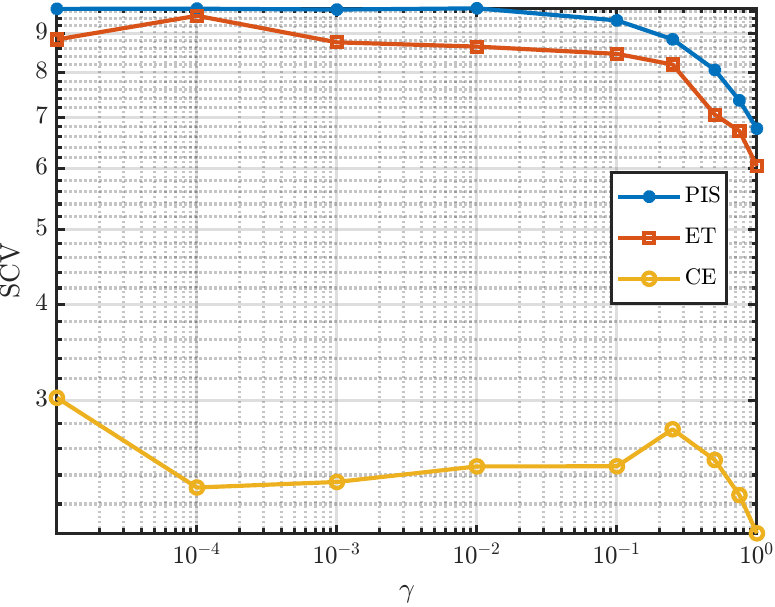}
        \caption{SCV versus $\gamma_{\mathrm{th}}$ for $M=8$, $m=4$, and $\mu=0.5$} 
        \label{fig:SCV_mu_half_m_4}
    \end{figure}

    \begin{figure}
        \centering
        \includegraphics[width=\linewidth]{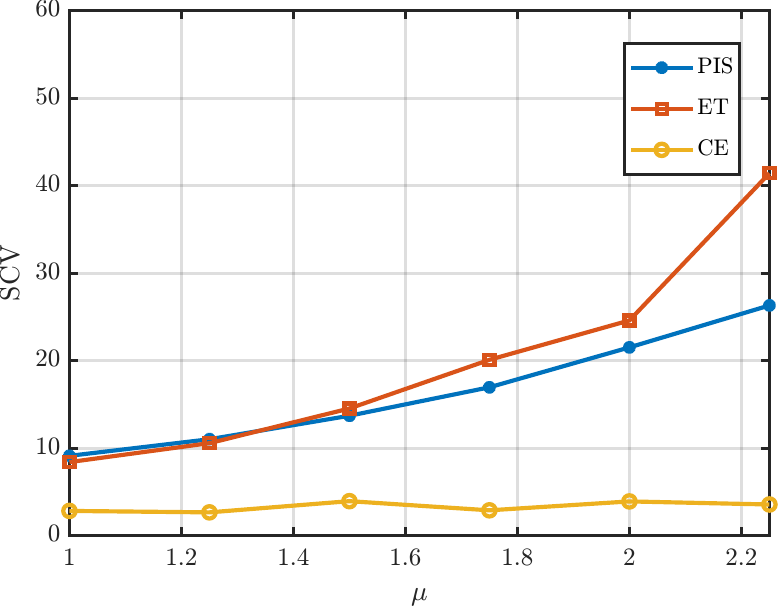}
        \caption{SCV versus $\mu$ for $M=8$, $m=4$, and $\gamma_{\mathrm{th}}=1$}
        \label{fig:SCV_gamma_1}
    \end{figure}

    \begin{figure}
        \centering
        \includegraphics[width=\linewidth]{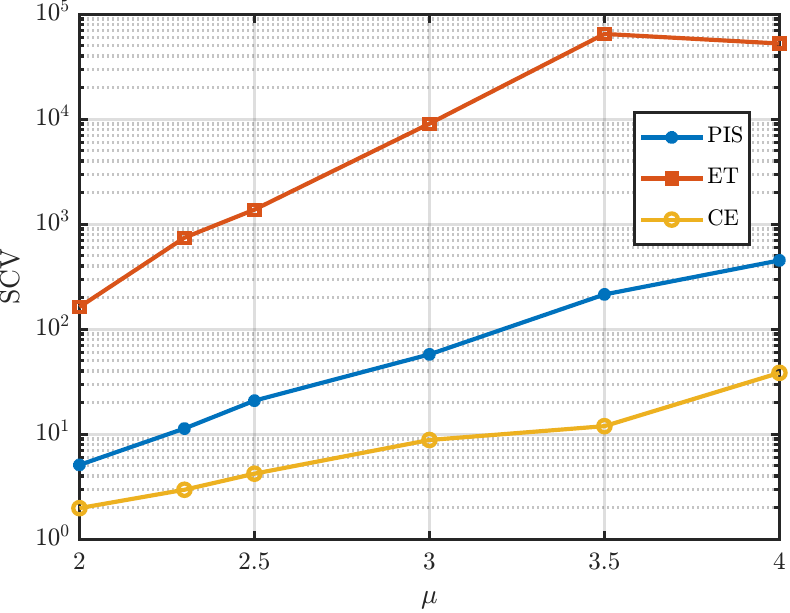}
        \caption{SCV versus $\mu$ for $M=8$, $m=4$, and $\gamma_{\mathrm{th}}=17$}
        \label{fig:SCV_gamma_17}
    \end{figure}

    Finally, note that for CE, the sufficient conditions on the asymptotic behaviour of the optimizer are observed in the implemented numerical applications. For example, Figure \ref{fig:ce-graph} shows the optimizer $[v_1,v_2]$ as a function of $\gamma_{\mathrm{th}}$ in a log-log scale for $M=8$, $m=4$, and $\mu=0.5$. We can see that the first parameter shows a line with slope $1$, and the second does not show significant change. Table \ref{tab:ce_v_params} shows the estimated probability, $v_2$, and $v_1/\gamma_{\mathrm{th}}$ as functions of $\gamma_{\mathrm{th}}$.

    \begin{figure}
        \centering
        \includegraphics[width=\linewidth]{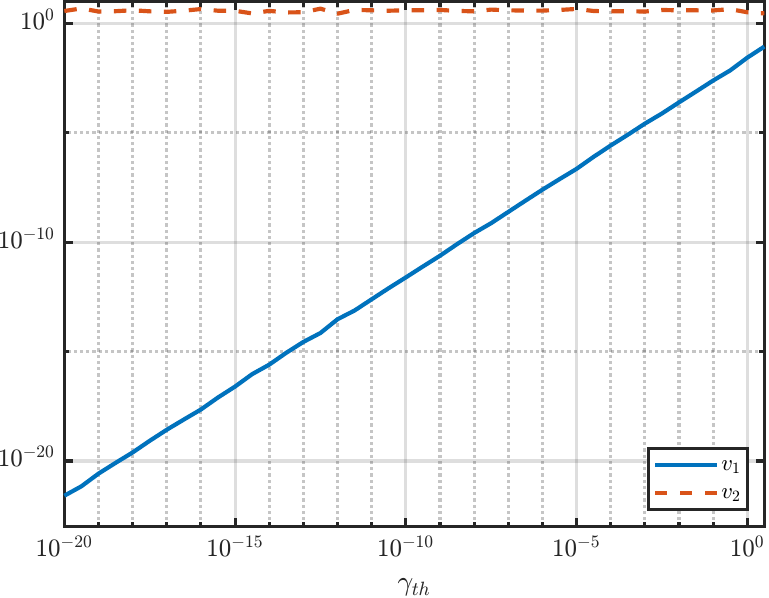}
        \caption{Parameters $v_1$ and $v_2$ as functions of $\gamma_{\mathrm{th}}$}
        \label{fig:ce-graph}
    \end{figure}

    \begin{table}[t]
        \centering
        \caption{CE Optimizer Parameters, $M=8$, $m=4$, $\mu=0.5$}
        \label{tab:ce_v_params}
        \resizebox{\columnwidth}{!}{\begin{tabular}{|c|c|c|c|c|}
        \hline
        $\gamma_{\mathrm{th}}$ & $\widehat{P}_{\mathrm{CE}}$ & WNRV & $v_1/\gamma_{\mathrm{th}}$ & $v_2$ \\
        \hline
        $10^{-19}$ & $2.2025\times 10^{-157}$ & $1.7833\times 10^{-5}$ & $2.6286\times 10^{-2}$ & $3.5085$ \\ \hline
        $10^{-18}$ & $2.2023\times 10^{-149}$ & $1.8985\times 10^{-5}$ & $2.4406\times 10^{-2}$ & $4.0020$ \\ \hline
        $10^{-17}$ & $2.2023\times 10^{-141}$ & $1.6689\times 10^{-5}$ & $2.6684\times 10^{-2}$ & $3.3877$ \\ \hline
        $10^{-16}$ & $2.2029\times 10^{-133}$ & $2.1735\times 10^{-5}$ & $2.2468\times 10^{-2}$ & $4.5476$ \\ \hline
        $10^{-15}$ & $2.2024\times 10^{-125}$ & $1.6715\times 10^{-5}$ & $2.5140\times 10^{-2}$ & $3.7770$ \\ \hline
        $10^{-14}$ & $2.2024\times 10^{-117}$ & $1.8954\times 10^{-5}$ & $2.5197\times 10^{-2}$ & $3.7572$ \\ \hline
        $10^{-13}$ & $2.2029\times 10^{-109}$ & $1.5875\times 10^{-5}$ & $2.7545\times 10^{-2}$ & $3.2455$ \\ \hline
        $10^{-12}$ & $2.2016\times 10^{-101}$ & $1.4907\times 10^{-5}$ & $2.9927\times 10^{-2}$ & $2.7855$ \\ \hline
        $10^{-11}$ & $2.2043\times 10^{-93}$  & $1.5377\times 10^{-5}$ & $2.4531\times 10^{-2}$ & $3.9868$ \\ \hline
        $10^{-10}$ & $2.2039\times 10^{-85}$  & $1.5115\times 10^{-5}$ & $2.4493\times 10^{-2}$ & $3.9856$ \\ \hline
        $10^{-9}$  & $2.2037\times 10^{-77}$  & $1.4771\times 10^{-5}$ & $2.3931\times 10^{-2}$ & $4.1488$ \\ \hline
        $10^{-8}$  & $2.2020\times 10^{-69}$  & $1.2944\times 10^{-5}$ & $2.6108\times 10^{-2}$ & $3.5763$ \\ \hline
        $10^{-7}$  & $2.2042\times 10^{-61}$  & $1.3771\times 10^{-5}$ & $2.4161\times 10^{-2}$ & $3.9754$ \\ \hline
        $10^{-6}$  & $2.2040\times 10^{-53}$  & $1.2569\times 10^{-5}$ & $2.5054\times 10^{-2}$ & $3.8632$ \\ \hline
        $10^{-5}$  & $2.2054\times 10^{-45}$  & $1.5415\times 10^{-5}$ & $2.2436\times 10^{-2}$ & $4.5930$ \\ \hline
        $10^{-4}$  & $2.2046\times 10^{-37}$  & $1.1536\times 10^{-5}$ & $2.6278\times 10^{-2}$ & $3.5685$ \\ \hline
        $10^{-3}$  & $2.2006\times 10^{-29}$  & $1.1367\times 10^{-5}$ & $2.6218\times 10^{-2}$ & $3.4909$ \\ \hline
        $10^{-2}$  & $2.1847\times 10^{-21}$  & $1.0968\times 10^{-5}$ & $2.4552\times 10^{-2}$ & $3.9747$ \\ \hline
        $10^{-1}$  & $2.0186\times 10^{-13}$  & $1.1245\times 10^{-5}$ & $2.4497\times 10^{-2}$ & $3.9271$ \\ \hline
        $10^{0}$   & $9.2192\times 10^{-6}$   & $8.3363\times 10^{-6}$ & $2.7191\times 10^{-2}$ & $3.2438$ \\
        \hline
        \end{tabular}}
    \end{table}

    \subsection{Case of Correlated Fading}

        Consider first the case of constant correlation. We test ETC and CEC for $M=8$, $m=4$, $\mu=0.5$, and $\rho=0.6$. The results are summarized in Table \ref{tab:variance-correlated-constant}.
        \begin{table*}[!t]
            \centering
            \caption{CEC vs ETC under constant correlation, $M=8$, $m=4$, $\mu=0.5$, $\rho=0.6$, $N=10^6$}
            \label{tab:variance-correlated-constant}
            \begin{tabular}{|c||c|c|c||c|c|c|}
                \hline
                 &\multicolumn{3}{c||}{ETC} &\multicolumn{3}{c|}{CEC}   \\ \hline
                 $\gamma_{\mathrm{th}}$ &$\hat{P}_N$ &RE (\%) &WNRV &$\hat{P}_N$ &RE (\%) &WNRV \\ \hline
                 $0.075$ &$1.072\times10^{-11}$ &$0.29$ &$2.8\times10^{-5}$
                 &$1.072\times10^{-11}$ & $0.24$ & $7.2\times10^{-6}$  \\ \hline
                 $0.1$ &$1.005\times10^{-10}$ &$0.28$ &$2.6\times10^{-5}$
                 &$1.007\times10^{-10}$ & $0.24$ & $6.0\times10^{-6}$  \\ \hline
                 $0.2$ &$1.997\times10^{-8}$ &$0.27$ &$2.2\times10^{-5}$
                 &$1.998\times10^{-8}$ & $0.23$ & $5.5\times10^{-6}$  \\ \hline
                 $0.3$ &$3.984\times10^{-7}$ &$0.26$ &$2.1\times10^{-5}$
                 &$3.989\times10^{-7}$ & $0.22$ & $5.0\times10^{-6}$  \\ \hline
                 $0.4$ &$3.105\times 10^{-6}$ &$0.25$ &$2.2\times 10^{-5}$
                 &$3.109\times 10^{-6}$ & $0.21$ & $5.3\times10^{-6}$  \\ \hline            
            \end{tabular}
        \end{table*} 
        To check the effect of correlation on the outage probability, we use the same example varying $\rho$ between $0$ and $0.8$. The results are plotted in Figure \ref{fig:rhos}. CEC was used, with samples of uniform size $N=10^6$, and the estimated relative error is less than $0.3\%$ in every estimate. Expectedly, as the correlation increases, the outage probability increases. We can see that the jumps between $\rho=0$ and $\rho=0.2$, and between $\rho=0.2$ and $\rho=0.4$ are comparable, while the gaps grow with larger correlations.

        \begin{figure}
            \centering
            \includegraphics[width=\linewidth]{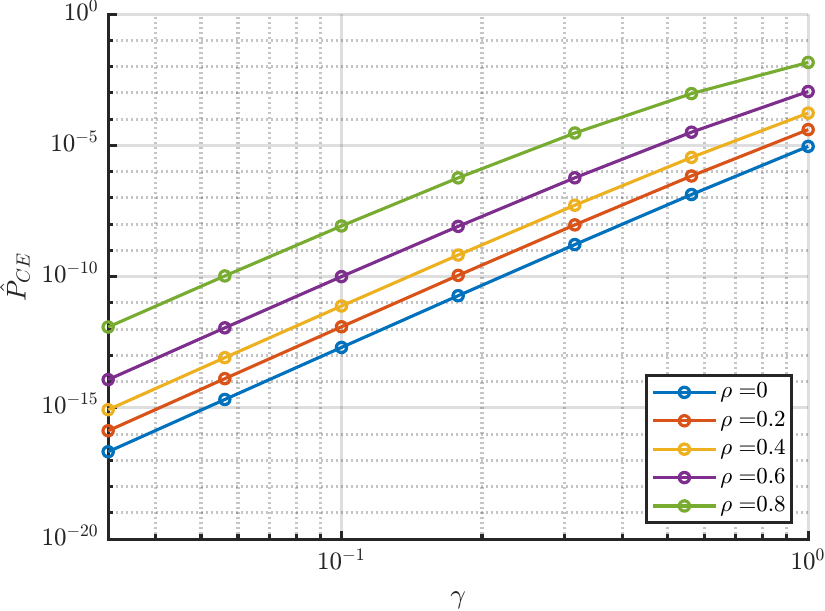}
            \caption{Outage Probability with different correlations as a function of $\gamma_{\mathrm{th}}$}
            \label{fig:rhos}
        \end{figure}

        Consider the case of a linear array of $M=8$ antennas, of which $m=4$ antennas are chosen, with correlated Rician fading in which the Rician parameters are given by $K_i=0.4$, $\Omega_i=1$, and the correlation is given by Jakes' model, also used in \cite{songAsymptoticAnalysisGSCRician2013}: \[\rho_{k\ell}=J_0(2\pi|k-\ell|d/\lambda)\] where $d$ is the separating distance between two adjacent antennas and $\lambda$ is the carrier wavelength. For $d=\lambda$, we give the results obtained by ETC and the asymptotic approximation in \cite{songAsymptoticAnalysisGSCRician2013} in Table \ref{tab:variance-correlated}. For the three larger probabilities, where naive Monte-Carlo (NMC) is computationally feasible, we also include NMC estimates as an independent check; these confirm ETC, not the asymptotic formula. The last column reports the relative difference between ETC and Song's approximation, defined as the absolute difference divided by ETC's estimate.

        \begin{table*}[!t]
            \centering
            \caption{ETC vs. Song's asymptotic approximation \cite{songAsymptoticAnalysisGSCRician2013} under Jakes' correlation, $M=8$, $m=4$, $K=0.4$, $\Omega=1$, $d=\lambda$}
            \label{tab:variance-correlated}
            \begin{tabular}{|c||c|c|c||c||c|c|c||c|}
                \hline
                 &\multicolumn{3}{c||}{ETC} &\multicolumn{1}{c||}{Song} &\multicolumn{3}{c||}{NMC} &  \\ \hline
                 $\gamma_{\mathrm{th}}$ &$\hat{P}_N$ &RE (\%) &WNRV &$\hat{P}$ &$\hat{P}_N$ &RE (\%) &WNRV & Rel Diff (\%) \\ \hline
                 $1$ &$6.802\times10^{-12}$ &$0.30$ &$3.3\times10^{-5}$
                 &$8.040\times10^{-12}$ & - & - & - &   18 \\ \hline
                 $2.5$ &$8.129\times10^{-9}$ &$0.30$ &$3.2\times10^{-5}$
                 &$1.227\times10^{-8}$ & - & - & - &   51 \\ \hline
                 $5$ &$1.393\times10^{-6}$ &$0.30$ &$3.6\times10^{-5}$
                 &$3.141\times10^{-6}$ & $1.450\times 10^{-6}$ & $8.3$  & $3.2\times 10^{-1}$ &  125 \\ \hline
                 $7.5$ &$2.407\times10^{-5}$ &$0.30$ &$4.8\times10^{-5}$
                 &$8.049\times10^{-5}$ & $2.412\times 10^{-5}$ & $2.0$ & $1.9\times10^{-2}$ &  234 \\ \hline
                 $10$ &$1.633\times 10^{-4}$ &$0.31$ &$6.9\times 10^{-5}$
                 &$8.040\times 10^{-4}$ & $1.643\times10^{-4}$ & $0.78$ & $2.6\times10^{-3}$ &  392 \\ \hline            
            \end{tabular}
        \end{table*} 

        As we can see, for moderately rare probabilities, Song's asymptotic approximation does not perform well. As we show in Figure \ref{fig:song-etc} which uses the same distribution, the asymptotic approximation is good for very small probabilities.

        \begin{figure}
            \centering
            \includegraphics[width=\linewidth]{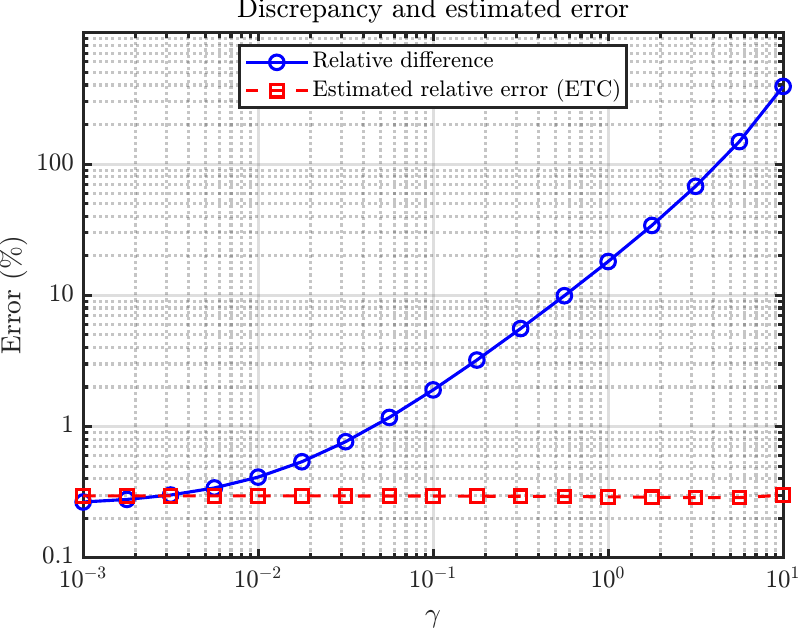}
            \caption{Relative Difference between Song's Asymptotic Formula and ETC Simulation Result}
            \label{fig:song-etc}
        \end{figure}

\section{Conclusion} \label{sec:conclusion}

    We developed and analysed a family of enhanced Monte Carlo estimators for the outage probability of GSC/MRC SIMO receivers under Rician fading. In the independent-fading case, we introduced partition importance sampling (PIS), and adapted exponential twisting (ET) and cross-entropy (CE) to the partial-sum-of-order-statistics problem; PIS and ET were shown to enjoy bounded relative error as the threshold tends to zero, and CE under an asymptotic condition on the optimizer. We benchmarked these methods against universal importance sampling \cite{benrachedSumOrderStatistics2018} and multilevel splitting \cite{benrachedUniversalSplittingEstimator2020}. In the correlated-fading setting we proposed ETC and CEC, where ETC is also shown to have bounded relative error.
    
    The numerical experiments paint a consistent picture. CE is the most robust method across the regimes we examined, despite its weaker theoretical guarantees. PIS is competitive but deteriorates as the line-of-sight components grow, because the acceptance-rejection sampling cost grows exponentially in the mean. ET deteriorates when the selected subset size is much smaller than the number of receive antennas, probably because its proposal targets the full sum rather than a partial one.  ETC is compared with a previously-established closed-form asymptotic approximation \cite{songAsymptoticAnalysisGSCRician2013}, and the comparison showed that the approximation may not be accurate for moderately rare events.
    
    Several directions remain open. First, the convergence of the CE iteration in our parametric family, and hence its BRE, should be established rigorously. Second, the CEC scheme as proposed assumes equal means and constant correlation; relaxing both would broaden its applicability. Third, the acceptance-rejection algorithm could be redesigned using a proposal density that fits the shape of the nominal one. Finally, extending multilevel splitting to correlated fading appears non-trivial because the construction of \cite{benrachedUniversalSplittingEstimator2020} hinges on monotone transforms of independent variables, and is worth a dedicated treatment.

\begin{appendices}

    \section{Proof of Proposition~\ref{prop:mode-bound}} \label{sec:proof-append-prop-1}

        Let \(\rs{X}\sim \chi_2^2(\lambda)\). Hence \[f_{\rs{X}}(x)=\dfrac{1}{2}\exp\left(-\dfrac{x+\lambda}{2}\right)I_0(\sqrt{\lambda x})\] Aubel et al. \cite{aubelUnimodalityNonCentral2000} proves that the distribution is strictly unimodal for $\lambda>2$. Let \(\ell(x)=\log f_{\rs{X}}(x)\). Then it can be shown that \begin{equation}
            \ell'(x)=-\frac12 +\dfrac{\sqrt{\lambda}}{2\sqrt{x}} \dfrac{I_1(\sqrt{\lambda x})}{I_0(\sqrt{\lambda x})}
        \end{equation} The mode $Z_\lambda$ is the unique solution of $\ell'(x)=0$, hence it satisfies \begin{equation}
            \dfrac{I_1(\sqrt{\lambda Z_\lambda})}{I_0(\sqrt{\lambda Z_\lambda})}=\dfrac{\sqrt{Z_\lambda}}{\sqrt{\lambda}}
        \end{equation} Equations (11) and (16) in \cite{amosComputationBessel1974} give \begin{equation}
            \dfrac{\sqrt{1+t^2}-1}{t}\leq \dfrac{I_1(t)}{I_0(t)}\leq \dfrac{\sqrt{1+4t^2}-1}{2t} \label{eq:bessel-bound}
        \end{equation} 
        Combining the last two equations, we can easily prove the proposition.

    \section{Proof of Proposition~\ref{prop:pdf-bound}}  \label{sec:proof-append-prop-2}

        Let \(\rs{X}\sim \chi_2^2(\lambda)\). Aubel et al. \cite{aubelUnimodalityNonCentral2000} proves that the corresponding PDF is log-concave.\footnote{Note that Karlin \cite{karlinTotalPositivity1964} proves that \(f_{\rs{X}}\) is a Polya-frequency density of order $2$, hence log-concave.} Let \(\ell(x)=\log f_{\rs{X}}(x)\). Then the function $\ell'$ is decreasing. Let $Z_\lambda$ be the mode. Let $x_0\in(\lambda-2,\lambda-1)$. Based on the relative position of the chosen $x_0$ and the mode, we consider two cases. 

        \textit{Case 1}: $x_0\leq Z_\lambda$. We can write \begin{align*}
            \ell(Z_\lambda)-\ell(x_0)&=\int_{x_0}^{Z_\lambda} \ell'(t)dt\leq (Z_\lambda-x_0)\ell'(x_0)\\
            &\leq (\lambda-1-x_0) \ell'(x_0)
        \end{align*} Using the upper bound in \eqref{eq:bessel-bound}, we can demonstrate that \begin{equation}
            \ell'(x_0)\leq -\dfrac{1}{2}+\dfrac{\sqrt{1+4\lambda x_0}-1}{4x_0}
        \end{equation}
        Therefore, \begin{equation}
            \log\left[\dfrac{f(Z_\lambda)}{f(x_0)}\right]\leq (\lambda-1-x_0)\left(-\dfrac{1}{2}+\dfrac{\sqrt{1+4\lambda x_0}-1}{4x_0}\right)
        \end{equation}

        \textit{Case 2}: $x_0 > Z_\lambda$. We can write \begin{equation}
            \ell(Z_\lambda)-\ell(x_0)\leq (x_0-\lambda+2)(-\ell'(x_0))
        \end{equation} Using the lower bound in \eqref{eq:bessel-bound}, and proceeding similarly, we can demonstrate that \begin{equation}
            \log\left[\dfrac{f(Z_\lambda)}{f(x_0)}\right]\leq(x_0-\lambda+2)\left(\dfrac{1}{2}-\dfrac{\sqrt{1+\lambda x_0}-1}{2x_0}\right)
        \end{equation}

    \section{Proof of Proposition~\ref{prop:Mell-explosion}} \label{proof-append-mean}

        Since $|\mu|$ goes to infinity, the threshold $\gamma_{\mathrm{th}}$ will eventually be smaller than the mode. Hence we can use \eqref{eq:Ml-large-mu-small-gamma}, which is rewritten here for the sake of convenience \[M_\ell=\dfrac{\left[2\gamma_{\mathrm{th}}f_{\chi_2^2(2|\mu|^2)}(2\gamma_{\mathrm{th}})\right]^m}{m!F_{\chi_{2m}^2(2m|\mu|^2)}(2\gamma_{\mathrm{th}})}.\]

        Firstly, we have \begin{equation*}
            2 f_{\chi_2^2(2|\mu|^2)}(2\gamma_{\mathrm{th}})=\exp\left(-\gamma_{\mathrm{th}}-|\mu|^2\right)I_0\left(2|\mu|\sqrt{\gamma_{\mathrm{th}}}\right).
        \end{equation*} We shall use \begin{equation}
            I_\nu(z)\underset{z\to\infty}{\sim}\dfrac{e^z}{\sqrt{2\pi z}}\left(1+\mathcal{O}\left(\dfrac{1}{z}\right)\right)
        \end{equation} from \cite{NIST:DLMF:10.40.1} to write \begin{equation}
            2 f_{\chi_2^2(2|\mu|^2)}(2\gamma_{\mathrm{th}})\sim\dfrac{e^{-\gamma_{\mathrm{th}}}}{2\sqrt{\pi}\gamma_{\mathrm{th}}^{\frac14}}\dfrac{e^{-|\mu|^2} e^{2|\mu|\sqrt{\gamma_{\mathrm{th}}}}}{|\mu|^{\frac12}}.
        \end{equation}

        On the other hand, we have \[F_{\chi_{2m}^2(2m|\mu|^2)}(2\gamma_{\mathrm{th}})=1-Q_m(\sqrt{2m}|\mu|,\sqrt{2\gamma_{\mathrm{th}}})\] where $Q_m(.,.)$ is the Marcum-Q function. 
        
        From \cite{temmeAsymptoticNoncentralChiSquare1993}, it holds that \begin{equation}
            Q_\nu(a,b)\sim 1-\dfrac{1}{2}\left(\dfrac{b}{a}\right)^{\nu-\frac{1}{2}}\operatorname{erfc}\left(\dfrac{a-b}{\sqrt{2}}\right)
        \end{equation} for $a>b$ and $ab\to\infty$. Hence, employing the fact that \begin{equation*}
            \operatorname{erfc}(x)=\dfrac{e^{-x^2}}{x\sqrt{\pi}}\left(1+\mathcal{O}\left(\dfrac{1}{x^2}\right)\right),
        \end{equation*} we can write \begin{equation}
            F_{\chi_{2m}^2(2m|\mu|^2)}\sim \dfrac{\gamma_{\mathrm{th}}^{(2m-1)/4}e^{-\gamma_{\mathrm{th}}}}{2\sqrt{\pi}m^{\frac{2m+1}{4}}} \dfrac{e^{-m|\mu|^2}e^{2\sqrt{m\gamma_{\mathrm{th}}}|\mu|}}{|\mu|^{m+\frac12}}
        \end{equation} as $|\mu|$ goes to infinity. Therefore, we have \eqref{eq:bound-asymp}, and the result is proved.
        
        \begin{figure*}[!t]
            \begin{equation}
                M_\ell\sim \dfrac{m^{\frac{2m+1}{4}}\gamma_{\mathrm{th}}^{\frac{m+1}{4}}}{2^{m-1}m!\pi^{\frac{m-1}{2}}\exp\left((m-1)\gamma_{\mathrm{th}}\right)}|\mu|^{\frac{m+1}{2}}\exp\left(2\sqrt{\gamma_{\mathrm{th}}}(m-\sqrt{m})|\mu|\right) \text{ as } |\mu| \to \infty
                \label{eq:bound-asymp}
            \end{equation}
            \hrulefill
        \end{figure*}

    \section{Proof of Proposition~\ref{prop:opt-ell2}} \label{sec:proof-opt-ell2}

        Assume $r\neq 0$. In the case of divisibility of $M$ by $m$, the proof is very similar. To formalize the problem, we set the following notations. Let $\mathcal{P}=(G_1,\ldots,G_q;R)$ be a partition of $\{1,\ldots,M\}$, so that $|G_t|=m$, for $t=1,\ldots,q$, and $|R|=r$. Basically, we think of these sets as the sets of indices in the partition represented by an arbitrary event $\mathcal{N}_2$. Let \[\Delta_t=\sum_{i\in G_t}|\mu_i|^2,\] and \[\Delta_R=\sum_{i\in R}|\mu_i|^2.\] The canonical partition, that corresponds to our sought result, is given by $\mathcal{P^*}=(G_1^*,\ldots,G_q^*;R^*)$, where $R^*=\{1,\ldots,r\}$, and $G_t^*=\{r+1+(t-1)m,\ldots,r+tm\}$, for $t=1,\ldots,q$. Let $H_k(\delta,\gamma_{\mathrm{th}})=F_{\chi_{2k}^2(2\delta)}(2\gamma_{\mathrm{th}})$, with $k=m$, or $k=r$. This function is log-concave in $\delta$ \cite{sunMontonoicityLogCooncavityMarcum2010}. Denote by $\dv{\Delta}(\mathcal{P})$ the vector of non-centrality parameters in $\ell_2$: $\dv{\Delta}(\mathcal{P})=(\Delta_q,\Delta_{q-1},\ldots,\Delta_1,\Delta_R)$, and write $\dv{\Delta}^*=(\Delta_q^*,\Delta_{q-1}^*,\ldots,\Delta_1^*,\Delta_R^*)$. Let \[\psi(\dv{\Delta})=\log \ell_2(\mathcal{P},\gamma_{\mathrm{th}})=\varphi_r(\Delta_R)+\sum_{t=1}^q \varphi_m(\Delta_t),\] where $\varphi_m(\delta)=\log H_m(\delta,\gamma_{\mathrm{th}})$, and $\varphi_r(\delta)=\log H_r(\delta,\gamma_{\mathrm{th}})$. Consider Proposition H.2 and its extension H.2.a in \cite{marshallInequalitiesMajorization1979}, flipping the inequalities in both the condition and the result in H.2.a, i.e., Schur-concavity instead of Schur-convexity. We can verify that, for any partition $\mathcal{P}$, $\dv{\Delta}(\mathcal{P})$ is majorized by $\dv{\Delta}^*$: $\dv{\Delta}(\mathcal{P})\prec \dv{\Delta}^*$. This holds because the sum of the first $t$ components of $\dv{\Delta}$ is less than that of the first largest $mt$ mean norms. To prove that the canonical partition $\mathcal{P}^*$ is a minimizer of $\ell_2$, it's sufficient to show that \begin{equation}
            \varphi_m'(a)\leq \varphi_m'(b) \text{ whenever } a\geq b, \label{eq:maj-1}
        \end{equation} and \begin{equation}
            \varphi_m'(a)\leq \varphi_r'(b) \text{ whenever } a\geq b. \label{eq:maj-2}
        \end{equation}
        A closer look at \ref{eq:maj-1} reveals that it's equivalent to a decreasing derivative, which is guaranteed by log-concavity. Now comes the second part. We can see that it's sufficient to prove that $\varphi_m'(b)\leq \varphi_r'(b)$ for $b\geq 0$. We can demonstrate that \[\varphi_m'(\delta)=\dfrac{1-Q_{m+1}(\sqrt{2\delta},\sqrt{2\gamma_{\mathrm{th}}})}{1-Q_{m}(\sqrt{2\delta},\sqrt{2\gamma_{\mathrm{th}}})}-1.\] Hence it is sufficient to prove that the fraction in the above equation is a decreasing function in $m$. But this is guaranteed by the log-concavity of the function $\nu \mapsto 1-Q_\nu(a, b)$ given again in \cite{sunMontonoicityLogCooncavityMarcum2010}. Hence the result is proved.

    \section{Proof of Proposition~\ref{prop:ET-BRE}} \label{sec:proof-append-1}

        We consider again the integer division of $M$ by $m$ \begin{equation*}
            M=qm+r;\;r\in\{0,1,\ldots,m-1\}
        \end{equation*} Hereafter, sums (resp. products) with lower index higher than upper index are considered to be empty and equal to zero (resp. one). Start with bounding the second moment of the estimate $\hat{\ell}_{ET}$ that can be written as \begin{equation*}
            \begin{aligned}
                &\mathbb{E}[\hat{\ell}_{\mathrm{ET}}^2]=\dfrac{\gamma_{\mathrm{th}}^{2M}\exp\left(-2\sum_{i=1}^M|\mu_i|^2\right)}{M^{2M}}\\
                &\times \mathbb{E}\left[\exp\left(2\dfrac{M-\gamma_{\mathrm{th}}}{\gamma_{\mathrm{th}}}\sum_{i=1}^M \rs{X}_i\right)\prod_{i=1}^M I_0^2(2|\mu_i|\sqrt{\rs{X}_i})\mathbf{1}_{\{H(\rv{X})\leq \gamma_{\mathrm{th}}\}}\right].
            \end{aligned}
        \end{equation*} Firstly, we utilize the simple fact \[\mathlarger{\mathbf{1}}_{\left\{H(\rv{X})\leq\gamma_{\mathrm{th}}\right\}}\leq \left[\prod_{t=0}^{q-1}\mathlarger{\mathbf{1}}_{\left\{\sum_{i=mt+1}^{mt+m} \rs{X}_{i}\leq \gamma_{\mathrm{th}}\right\}}\right]\mathlarger{\mathbf{1}}_{\left\{\sum_{i=mq+1}^M \rs{X}_{i}\leq \gamma_{\mathrm{th}}\right\}}.\] Hence we have \eqref{eq:first-bound}.  
        \begin{figure*}[!t]
            \begin{equation}
                \begin{aligned}
                    &\mathbb{E}[\hat{\ell}_{\mathrm{ET}}^2]\leq \dfrac{\gamma_{\mathrm{th}}^{2M}\exp\left(-2\sum_{i=1}^M|\mu_i|^2\right)}{M^{2M}} \prod_{t=0}^{q-1} \mathbb{E}\left[\exp\left(2\dfrac{M-\gamma_{\mathrm{th}}}{\gamma_{\mathrm{th}}}\sum_{i=mt+1}^{mt+m}\rs{X}_i\right)\left[\prod_{i=mt+1}^{mt+m}I_0^2(2|\mu_i|\sqrt{\rs{X}_i})\right]\mathbf{1}_{\left\{\sum_{i=mt+1}^{mt+m}\rs{X}_i\leq\gamma_{\mathrm{th}}\right\}}\vphantom{\exp\left(2\dfrac{M-\gamma_{\mathrm{th}}}{\gamma_{\mathrm{th}}}\sum_{i=mt+1}^{mt+m}\rs{X}_i\right)}\right]\\
                    &\times \mathbb{E}\left[\exp\left(2\dfrac{M-\gamma_{\mathrm{th}}}{\gamma_{\mathrm{th}}}\sum_{i=mq+1}^{M}\rs{X}_i\right)\prod_{i=mq+1}^{M}\left[I_0^2(2|\mu_i|\sqrt{\rs{X}_i})\right]\mathbf{1}_{\left\{\sum_{i=mq+1}^{M}\rs{X}_i\leq\gamma_{\mathrm{th}}\right\}}\vphantom{\exp\left(2\dfrac{M-\gamma_{\mathrm{th}}}{\gamma_{\mathrm{th}}}\sum_{i=mq+1}^{M}\rs{X}_i\right)}\right]\\
                \end{aligned} \label{eq:first-bound}
            \end{equation}
            \hrulefill
        \end{figure*}

        Since under the sampling distribution, $X_i$ are identically distributed, any indexing independent of $\mu_i$ can be reset to start from $1$; this may clear the clutter a bit. Now since \begin{equation*}
            I_0(x)=\dfrac{1}{\pi}\int_0^\pi e^{x\cos\theta}d\theta\leq\dfrac{1}{\pi}\int_0^\pi e^xd\theta=e^x
        \end{equation*} for any positive real number $x$, we can write \eqref{eq:second-bound}. \begin{figure*}[!t]
            \begin{equation}
                \begin{aligned}
                    &\mathbb{E}[\hat{\ell}_{\mathrm{ET}}^2]\leq \dfrac{\gamma_{\mathrm{th}}^{2M}\exp\left(-2\sum_{i=1}^M|\mu_i|^2\right)}{M^{2M}} \prod_{t=0}^{q-1} \mathbb{E}\left[\exp\left(2\dfrac{M-\gamma_{\mathrm{th}}}{\gamma_{\mathrm{th}}}\sum_{i=1}^{m}\rs{X}_i\right)\exp\left(4\sum_{i=mt+1}^{mt+m}|\mu_i|\sqrt{\rs{X}_i}\right)\mathbf{1}_{\left\{\sum_{i=1}^{m}\rs{X}_i\leq\gamma_{\mathrm{th}}\right\}}\right]\\
                    &\times \mathbb{E}\left[\exp\left(2\dfrac{M-\gamma_{\mathrm{th}}}{\gamma_{\mathrm{th}}}\sum_{i=1}^{r}\rs{X}_i\right)\exp\left(4\sum_{i=mt+1}^{mt+m}|\mu_i|\sqrt{\rs{X}_i}\right)\mathbf{1}_{\left\{\sum_{i=1}^{r}\rs{X}_i\leq\gamma_{\mathrm{th}}\right\}}\vphantom{\exp\left(2\dfrac{M-\gamma_{\mathrm{th}}}{\gamma_{\mathrm{th}}}\sum_{i=mq+1}^{M}\rs{X}_i\right)}\right]\\
                \end{aligned} \label{eq:second-bound}
            \end{equation}
            \hrulefill
        \end{figure*}
        By Cauchy-Schwarz inequality, we have \begin{equation*}
            \sum_{i=1}^{m}|\mu_i|\sqrt{\rs{X}_i}\leq \left(\sum_{i=1}^m|\mu_i|^2\right)^{1/2}\left(\sum_{i=1}^m\rs{X}_i\right)^{1/2},
        \end{equation*} the equation which can be evidently replicated for different indices. Then we have \eqref{eq:third-bound}. \begin{figure*}[!t]
                    \begin{equation}
                        \begin{aligned}
                            &\mathbb{E}[\hat{\ell}_{\mathrm{ET}}^2]\leq \dfrac{\gamma_{\mathrm{th}}^{2M}\exp\left(-2\sum_{i=1}^M|\mu_i|^2\right)}{M^{2M}}\\
                            &\times\prod_{t=0}^{q-1} \mathbb{E}\left[\exp\left(2\dfrac{M-\gamma_{\mathrm{th}}}{\gamma_{\mathrm{th}}}\sum_{i=1}^{m}\rs{X}_i\right)\exp\left\{4\left(\sum_{i=mt+1}^{mt+m}|\mu_i|^2\right)^{1/2}\left(\sum_{i=1}^m\rs{X}_i\right)^{1/2}\right\}\mathbf{1}_{\left\{\sum_{i=1}^{m}\rs{X}_i\leq\gamma_{\mathrm{th}}\right\}}\right]\\
                            &\times \mathbb{E}\left[\exp\left(2\dfrac{M-\gamma_{\mathrm{th}}}{\gamma_{\mathrm{th}}}\sum_{i=1}^{r}\rs{X}_i\right)\exp\left\{4\left(\sum_{i=mq+1}^{M}|\mu_i|^2\right)^{1/2}\left(\sum_{i=1}^r\rs{X}_i\right)^{1/2}\right\}\mathbf{1}_{\left\{\sum_{i=1}^{r}\rs{X}_i\leq\gamma_{\mathrm{th}}\right\}}\vphantom{\exp\left(2\dfrac{M-\gamma_{\mathrm{th}}}{\gamma_{\mathrm{th}}}\sum_{i=mq+1}^{M}\rs{X}_i\right)}\right]\\
                        \end{aligned} \label{eq:third-bound}
                    \end{equation}
                    \hrulefill
                \end{figure*}
        Now, since all the random variables appear in similar sums that are to be bounded in the indicator functions by $\gamma_{\mathrm{th}}$, and $\gamma_{\mathrm{th}}$ is eventually less than $M$, we can write \begin{equation*}
            \begin{aligned}
                &\mathbb{E}[\hat{\ell}_{\mathrm{ET}}^2]\leq \dfrac{\gamma_{\mathrm{th}}^{2M}\exp\left(-2\sum_{i=1}^M|\mu_i|^2\right)}{M^{2M}}\\
                &\times \prod_{t=0}^{q-1} \exp(2(M-\gamma_{\mathrm{th}}))\exp(4\sqrt{\gamma_{\mathrm{th}}}d_t)\mathbb{P}\left[\sum_{i=1}^m \rs{X}_i\leq \gamma_{\mathrm{th}}\right]\\
                &\times \exp(2(M-\gamma_{\mathrm{th}}))\exp(4\sqrt{\gamma_{\mathrm{th}}}d_q)\mathbb{P}\left[\sum_{i=1}^r\rs{X}_i\leq\gamma_{\mathrm{th}}\right],
            \end{aligned}
        \end{equation*} 
        where \[d_t=\left(\sum_{i=mt+1}^{mt+m}|\mu_i|^2\right)^{1/2}\] for $t=0,\ldots,q-1$, and \[d_q=\left(\sum_{i=mq+1}^{M}|\mu_i|^2\right)^{1/2}.\] Now, since $\rs{X}_i\sim\mathrm{Exp}(M/\gamma_{\mathrm{th}})\equiv\frac{\gamma_{\mathrm{th}}}{2M}\chi_2^2$, we have \begin{equation*}
            \sum_{i=1}^m \rs{X}_i \sim \frac{\gamma_{\mathrm{th}}}{2M}\chi_{2m}^2.
        \end{equation*} 
        Thus we can write\begin{equation*}
            \begin{aligned}
                &\mathbb{E}[\hat{\ell}_{\mathrm{ET}}^2]\leq \dfrac{\gamma_{\mathrm{th}}^{2M}e^{-2\|\dv{\mu}\|^2}}{M^{2M}}\exp(2qM-2q\gamma_{\mathrm{th}})\exp\left(4\sqrt{\gamma_{\mathrm{th}}}\sum_{t=0}^q d_t\right)\\
                &\times \left[1-e^{-M} \sum_{k=0}^{m-1} \frac{M^k}{k!}\right]^q\left[1-e^{-M} \sum_{k=0}^{r-1} \frac{M^k}{k!}\right].
            \end{aligned}
        \end{equation*} 
        On the other hand, as the sought $P$ is the probability of the maximum partial sum being less than $\gamma_{\mathrm{th}}$ is less than that of the whole sum, we have \begin{equation}
            P>F_{\chi_{2M}^2(2\|\dv{\mu}\|^2)}(2\gamma_{\mathrm{th}}),
        \end{equation}
        where $F_{\chi_{2M}^2}$ denotes the CDF of $\chi_{2M}^2$. Considering the mixture expansion of non-central chi-square in terms of (central) chi-squares, we can say that as $\gamma_{\mathrm{th}}$ goes to zero, we have \begin{equation*}
            \begin{aligned}
                F_{\chi_{2M}^2(2\|\dv{\mu}\|^2)}(2\gamma_{\mathrm{th}}) &\sim e^{-\|\dv{\mu}\|^2}\dfrac{\gamma(M,\gamma_{\mathrm{th}})}{\Gamma(M)}\\
                &\sim \dfrac{e^{-\|\dv{\mu}\|^2}\gamma_{\mathrm{th}}^M}{\Gamma(M+1)},
            \end{aligned}
        \end{equation*}
        as $\gamma_{\mathrm{th}}$ goes to zero, where $\gamma(.,.)$ is the lower incomplete gamma function. Squaring the last bound, and dividing the upper bound of the second moment of the estimator by the lower bound of $P$, $\gamma_{\mathrm{th}}^{2M}$ is cancelled out and our proof is completed.

            \section{Proof of Proposition~\ref{prop:ETC-BRE}} \label{sec:proof-BRE-ETC}
            
                Write first \begin{equation}
                    \begin{aligned}
                        \hat{\ell}_{ETC}^2&=\dfrac{\gamma_{\mathrm{th}}^{2M}}{M^{2M}|\dmt{\Sigma}|^{2}}\exp\left(\dfrac{2M}{\gamma_{\mathrm{th}}}\rv{h}^H\rv{h}\right.\\
                        &\left.-2(\rv{h}-\dv{\mu})^H\dmt{\Sigma}^{-1}(\rv{h}-\dv{\mu})\vphantom{\dfrac{\gamma_{\mathrm{th}}^{2M}}{M^{2M}|\dmt{\Sigma}|^{2}}}\right)\mathbf{1}_{\left\{G(\rv{h})\leq\gamma_{\mathrm{th}}\right\}}
                    \end{aligned}
                \end{equation} where \[G(\rv{h})=\max\sum_{\ell=1}^m|\rs{h}_{k_\ell}|^2\] Using the fact that \[\exp\left(-2(\rv{h}-\dv{\mu})^H\dmt{\Sigma}^{-1}(\rv{h}-\dv{\mu})\vphantom{\dfrac{\gamma_{\mathrm{th}}^{2M}}{M^{2M}|\dmt{\Sigma}|^{2M}}}\right)\leq 1\] and partitioning similarly to the previous cases we get \eqref{eq:ETC-1}. \begin{figure*}[!t]
                \begin{equation}
                    \begin{aligned}
                        \hat{\ell}_{\mathrm{ETC}}^2&\leq \left(\dfrac{\gamma_{\mathrm{th}}}{M}\right)^{2M} \dfrac{1}{|\dmt{\Sigma}|^2}\prod_{t=0}^{q-1} \left[\exp\left(\dfrac{2M}{\gamma_{\mathrm{th}}}\sum_{k=mt+1}^{mt+m}|\rs{h}_k|^2\right)\mathbf{1}_{\left\{\sum_{k=mt+1}^{mt+m}|\rs{h}_k|^2\leq \gamma_{\mathrm{th}}\right\}}\right]\\
                        &\times \exp\left(\dfrac{2M}{\gamma_{\mathrm{th}}}\sum_{k=mq+1}^{M}|\rs{h}_k|^2\right)\mathbf{1}_{\left\{\sum_{k=mq+1}^{M}|\rs{h}_k|^2\leq \gamma_{\mathrm{th}}\right\}}
                    \end{aligned} \label{eq:ETC-1}
                \end{equation}
                \hrulefill
            \end{figure*} Hence, using the fact that the indicator function is less than or equal to unity\footnote{In the independent case, we could have used a similar argument, but we chose to use a sharper bound. In the case of correlation, there is a closed form expression for such a probability \cite{al-naffouriDistributionIndefiniteQuadratic2016}, but it isn't as simple as the previous case.}, we have \begin{equation}
                \mathbb{E}[\hat{\ell}_{\mathrm{ETC}}^2]\leq \dfrac{\gamma_{\mathrm{th}}^{2M}\exp(2M(q+1))}{M^{2M}|\dmt{\Sigma}|^{2}} . 
            \end{equation} On the other hand, we have \begin{equation}
                P>\mathbb{P}[\|\rv{h}\|^2\leq \gamma_{\mathrm{th}}]
            \end{equation} Write \[\rv{h}=\dv{\mu}+\dmt{\Sigma}^{1/2}\rv{z}\] where $\rv{z}\sim\mathcal{CN}(\dv{0},\dmt{I}_M)$. Then \[\|\rv{h}\|^2=(\tilde{\dv{\mu}}+\rv{z})^H\dmt{\Sigma}(\tilde{\dv{\mu}}+\rv{z})\leq \lambda_1 \|\tilde{\dv{\mu}}+\rv{z}\|^2\] where $\tilde{\dv{\mu}}=\dmt{\Sigma}^{-1/2}\dv{\mu}$ and $\lambda_1$ is the maximum eigenvalue of $\dmt{\Sigma}$. Hence we have \begin{equation}
                \begin{aligned}
                    P &\geq \mathbb{P}\left[\lambda_1 \|\tilde{\dv{\mu}}+\rv{z}\|^2 \leq \gamma_{\mathrm{th}} \right]\\
                    &=F_{\chi_{2M}^2(2\|\tilde{\dv{\mu}}\|^2)}\left(\dfrac{2\gamma_{\mathrm{th}}}{\lambda_1}\right)\sim \dfrac{e^{-\|\tilde{\dv{\mu}}\|^2}\gamma_{\mathrm{th}}^M}{M! \lambda_1^M}
                \end{aligned} 
            \end{equation} Therefore the result holds.

\end{appendices}

\bibliographystyle{ieeetr} 
\bibliography{BibFiles/BibFinal2}

\end{document}